%% file: RandomHeston.tex
\documentclass[11pt,a4paper]{article}

\input{preamble}

\begin{document}

\graphicspath{{graphics/}}

\title{Stationary Heston model: Calibration and Pricing of exotics using Product Recursive Quantization}
\author{
	\sc Vincent Lemaire
	\thanks{Sorbonne Université, Laboratoire de Probabilités, Statistique et Modélisation, LPSM, Campus Pierre et Marie Curie, case 158, 4 place Jussieu, F-75252 Paris Cedex 5, France.}
	\and
	\sc Thibaut Montes
	\samethanks[1]
	\thanks{The Independent Calculation Agent, The ICA, 112 Avenue Kleber, 75116 Paris, France.}
	\and
	\sc Gilles Pagès
	\samethanks[1]
}
\maketitle

\begin{abstract}
	A major drawback of the Standard Heston model is that its implied volatility surface does not produce a steep enough smile when looking at short maturities. For that reason, we introduce the Stationary Heston model where we replace the deterministic initial condition of the volatility by its invariant measure and show, based on calibrated parameters, that this model produce a steeper smile for short maturities than the Standard Heston model. We also present numerical solution based on Product Recursive Quantization for the evaluation of exotic options (Bermudan and Barrier options).
\end{abstract}


\input{corpus}

\section*{Acknowledgment}
The authors wish to thank Guillaume Aubert for fruitful discussion on the Heston model and Jean-Michel Fayolle for his advice on the calibration of the models. The PhD thesis of Thibaut Montes is funded by a CIFRE grand from The Independent Calculation Agent (The ICA) and French ANRT.

\newpage

\nocite{*}
\bibliography{bibli}
\bibliographystyle{alpha}

\newpage

\begin{appendices}

\input{appendices}

\end{appendices}

\end{document}

%% file: preamble.tex
\usepackage[utf8]{inputenc}
\usepackage[english]{babel}
\usepackage[T1]{fontenc}
\usepackage[left=2.5cm,right=2.5cm,top=3cm,bottom=3cm]{geometry}
\usepackage[toc,page]{appendix}
\usepackage{amsfonts,amssymb,amsthm,graphicx,dsfont,xcolor,stmaryrd,subcaption}
\usepackage{amsmath, mathtools}
\usepackage{enumitem}
\usepackage{hyperref}
\usepackage{float} 
\usepackage{multirow}
\usepackage[normalem]{ulem} 
\usepackage{array, booktabs}

\usepackage{mathabx}

\DeclareMathOperator{\e}{e}

\DeclareMathOperator{\Proj}{\mathrm{Proj}}
\DeclareMathOperator{\dist}{\mathrm{dist}}

\DeclareMathOperator{\A}{\mathcal{A}}
\DeclareMathOperator{\F}{\mathcal{F}}
\DeclareMathOperator{\N}{\mathcal{N}}

\DeclareMathOperator{\R}{\mathds{R}}
\DeclareMathOperator{\Integer}{\mathds{N}}

\DeclareMathOperator{\E}{\mathds{E}}

\DeclareMathOperator{\Corr}{\mathds{C}\mathrm{orr}}
\DeclareMathOperator{\Prob}{\mathds{P}}
\DeclareMathOperator{\Law}{\mathcal{L}}
\DeclareMathOperator{\Tr}{Tr}

\DeclareMathOperator*{\argmin}{arg\,min}

\DeclareMathOperator{\1}{\mathds{1}}

\DeclareMathOperator{\Distortion}{\mathcal{Q}_{2,N}}

\newcommand{\vertiii}[1]{{\left\vert\kern-0.25ex\left\vert\kern-0.25ex\left\vert #1
\right\vert\kern-0.25ex\right\vert\kern-0.25ex\right\vert}}

\theoremstyle{plain}
\newtheorem{theorem}{Theorem}[section]
\newtheorem{lemme}[theorem]{Lemma}
\newtheorem{proposition}[theorem]{Proposition}

\theoremstyle{definition}
\newtheorem{definition}[theorem]{Definition}

\newtheorem{example}[theorem]{Example}

\newtheorem{remark}[theorem]{\noindent \textbf{Remark}}
\newtheorem{remarks}[theorem]{\noindent \textbf{Remarks}}

\numberwithin{equation}{section}

\renewenvironment{abstract}
 {\small
	\begin{center}
		\bfseries \abstractname\vspace{-.5em}\vspace{0pt}
	\end{center}
	\list{}{%
	\setlength{\leftmargin}{20mm}
	\setlength{\rightmargin}{\leftmargin}%
 }%
 \item\relax}
 {\endlist}

\newcommand*\samethanks[1][\value{footnote}]{\footnotemark[#1]}

%% file: corpus.tex
\section*{Introduction}

Originally introduced by Heston in \cite{heston1993closed}, the Heston model is a stochastic volatility model used in Quantitative Finance to model the joint dynamics of a stock and its volatility, denoted $(S_t^{(x)})_{t \geq 0}$ and $(v_t^x)_{t \geq 0}$, respectively, where $v_0^x = x$ is the initial condition of the volatility. Historically, the initial condition of the volatility $x$ is considered as deterministic and is calibrated in the market like the other parameters of the model. This model received an important attention among practitioners for two reasons: first, it is a stochastic volatility model, hence it introduces smile in the implied volatility surface as observed in the market, which is not the case of models with constant volatility and second, in its original form, we have access to a semi closed-form formula for the characteristic function which allows us to price European options (Call $\&$ Put) almost instantaneously using the Fast Fourier approach (Carr $\&$ Madan in \cite{carr1999option}). Yet, a complaint often heard about the Heston model is that it fails to fit the implied volatility surface for short maturities because the model cannot produce a steep-enough smile for those maturities (see \cite{gatheral2011volatility}).

Noticing that the volatility process is ergodic with a unique invariant distribution $\nu = \Gamma (\alpha, \beta)$ where the parameters $\alpha$ and $\beta$ depend on the volatility diffusion parameters, it has been first proposed by Pagès $\&$ Panloup in \cite{pages2009approximation} to directly consider that the process evolves under its stationary regime in place of starting it at time $0$ from a deterministic value. We denote by $(S_t^{(\nu)})_{t \geq 0}$ and $(v_t^{\nu})_{t \geq 0}$ the couple asset-volatility in the Stationary Heston model. Replacing the initial condition of the volatility by the stationary measure does not modify the long-term behavior of the implied volatility surface but does inject more randomness into the volatility for short maturities. This tends to produce a steeper smile for short maturities, which is the kind of behavior we are looking for. Later, the short-time and long-time behavior of the implied volatility generated by such model has been studied by Jacquier $\&$ Shi in \cite{jacquier2017randomized}. Another extension of the Heston model have been suggested and extensivel analyzed in order to reproduce the slope of the skew for short-term expiring options: the Rough Heston model where the volatility satisfies a Voltera equation driven by a "rough" Liouville process with $H$-Hölder paths, $H=0.1$ (see \cite{jaisson2016rough,guennoun2018asymptotic,gatheral2018volatility,giorgia2018fast,gatheral2019rational} for details on the model and numerical solutions).

An other extension of the Heston model have been suggested in order to be able to reproduce the slope of the skew for short-term expiring options: the Rough Heston model (see \cite{jaisson2016rough,guennoun2018asymptotic,gatheral2018volatility,giorgia2018fast,gatheral2019rational} for details on the model and numerical solutions).

\medskip
In the beginning of the paper, we briefly recall the well-known methodology used for the pricing of European option in the Standard Heston model. Based on that, we express the price $I_0$ of a European option on the asset $S_T^{(\nu)}$ as
\begin{equation}
	I_0 = \E \big[ \e^{-r T} \varphi ( S_T^{(\nu)} ) \big] = \E \big[ f (v_0^{\nu}) \big]
\end{equation}
where $f(v)$ is the price of the European option in the Standard Heston model for a given set of parameters. The last expectation can be computed efficiently using quadrature formulas either based on optimal quantization of the Gamma distribution or on Laguerre polynomials.

Once we are able to price European options, we can think of calibrating our model to market data. Indeed the parameters of the model are calibrated using the implied volatility surface observed in the market. However, the calibration of the Standard Heston model is highly depending on the initial guess we choose in the minimization problem. This is due to an over-parametrization of the model (see \cite{gauthier2009fitting}). Hence, when we consider the Heston model in its stationary regime, there is one parameter less to calibrate as the initial value of the volatility is no longer deterministic. The stationary model tends to be more robust when it comes to calibration.

\medskip

In the second part of paper, we deal with the pricing of Exotic options such as Bermudan and Barrier options. We propose a method based on hybrid product recursive quantization. The "hybrid" term comes from the fact that we use two different types of schemes for the discretization of the volatility and the asset (Milstein and Euler-Maruyama). Recursive quantization was first introduced by Pagès $\&$ Sagna in \cite{pages2015recursive}. It is a Markovian quantization (see \cite{pages2004optimal}) drastically improved by the introduction of fast deterministic optimization procedure of the quantization grids and the transition weights. This optimization allows them to drastically reduce the time complexity by an order of magnitude and build such trees in a few seconds. Originally devised for Euler-Maruyama scheme of one dimensional Brownian diffusion, it has been extended to one-dimensional higher-order schemes by \cite{mcwalter2018recursive} and to medium dimensions using product quantization (see \cite{abbas2018product,rudd2017fast,callegaro2017american,callegaro2017pricing,sagna2018general}). Then, once the quantization tree is built, we proceed by a backward induction using the Backward Dynamic Programming Principle for the price of Bermudan options and using the methodology detailed in \cite{sagna2010pricing,pages2018numerical} based on the conditional law of the Brownian Bridge for the price of Barrier options.

\medskip

The paper is organized as follows. First, in Section \ref{RH:section:the_model}, we recall the definition of the Standard Heston model and the interesting features of the volatility diffusion which bring us to define the Stationary Heston model. In Section \ref{RH:section:pricing_calibrationEU}, we give a fast solution for the pricing of European options in the Stationary Heston model when there exists methods for the Standard model. Finally, once we are able to price European options, we can define the optimization problem of calibration on implied volatility surface. We perform the calibration of both models and compare their induced smile for short maturities options. Once this model has been calibrated, in Section \ref{RH:section:exotic_options}, we propose a numerical method based on hybrid product recursive quantization for the pricing of exotic financial products: Bermudan and Barrier options. For this method, we give an estimate of the $L^2$-error introduced by the approximation.

\section{The Heston Model} \label{RH:section:the_model}
The Standard Heston model is a two-dimensional diffusion process $(S_t^{(x)}, v_t^x)$ solution to the Stochastic Differential Equation
\begin{equation} \label{RH:model_standard_heston}
	\left\{
	\begin{aligned}
		\frac{dS_t^{(x)}}{S_t^{(x)}} & = (r-q) dt + \sqrt{v_t^x} \big( \rho d \widetilde W_t + \sqrt{1 - \rho^2} d W_t \big) \\
		dv_t^x                       & = \kappa (\theta - v_t^x) dt + \xi \sqrt{v_t^x} d \widetilde W_t
	\end{aligned}
	\right.
\end{equation}
where
\begin{itemize}
	\item $S_t^{(x)}$ is the dynamic of the risky asset,
	\item $v_t^x$ is the dynamic of the volatility process,
	\item $S_0^{(x)} = s_0 \geq 0$ is the initial value of the process,
	\item $r \in \R$ denotes the interest rate,
	\item $q \in \R$ is the dividend rate,
	\item $\rho \in [-1, 1]$ is the correlation between the asset and the volatility,
	\item $(W, \widetilde W)$ is a two-dimensional standard Brownian motion,
	\item $\theta \geq 0$ the long run average price variance,
	\item $\kappa \geq 0$ the rate at which $v_t^x$ reverts to $\theta$,
	\item $\xi \geq 0$ is the volatility of the volatility,
	\item $v_0^x = x \geq 0$ is the deterministic initial condition of the volatility.
\end{itemize}

This model is widely used by practitioner for various reasons. One is that it leads to semi-closed forms for vanilla options based on a fast Fourier transform. The other is that it represents well the observed mid and long-term market behavior of the implied volatility surface observed on the market. However, it fails producing or even fitting to the smile observed for short-term maturities.

\begin{remark}[The volatility] \label{RH:remark::vol}

	One can notice that the volatility process is autonomous thence we are facing a one dimensional problem. Moreover, the volatility process is following a Cox-Ingersoll-Ross (CIR) diffusion also known as the square root diffusion.
	Existence and uniqueness of a strong solution to this stochastic differential equation has been first shown in \cite{ikeda1981stochastic}, if $x \geq 0$. Moreover, it has been shown, see \cite{lamberton2011introduction}, that if the Feller condition holds, namely $\xi^2 \leq 2 \kappa \theta$, for every $x > 0$, then the unique solution $(v_t^x)_{t \geq 0}$ satisfies
	\begin{equation}
		\forall t \geq 0, \quad \Prob ( \tau_0^x = + \infty ) = 1
	\end{equation}
	where $\tau_0^x$ is the first hitting time defined by
	\begin{equation}
		\tau_0^x = \inf \{ t \geq 0 \mid v_t^x = 0 \} \quad \mbox{ where } \inf \emptyset = + \infty.
	\end{equation}
	Moreover, the CIR diffusion admits, as a Markov process, a unique stationary regime, characterized by its invariant distribution
	\begin{equation}\label{RH:gammalaw}
		\nu = \Gamma(\alpha, \beta)
	\end{equation}
	where
	\begin{equation}\label{RH:params_gammalaw}
		\alpha = \theta \beta \quad \mbox{ and }\quad \beta = 2 \kappa / \xi^2.
	\end{equation}

\end{remark}

\medskip
Based on the above remarks, the idea is to precisely consider the volatility process under its stationary regime, i.e., replacing the deterministic initial condition from the Standard Heston model by a $\nu$-distributed random variable independent of $(W, \widetilde W)$. We will refer to this model as the Stationary Heston model. Our first aim is to inject more randomness for short maturities ($t$ small) into the volatility but also to reduce the number of free parameters to stabilize and robustify the calibration of the Heston model which is commonly known to be overparametrized (see e.g. \cite{gauthier2009fitting}).

\smallskip
This model was first introduced by \cite{pages2009approximation} (see also \cite{ikeda1981stochastic}, p. 221). More recently, \cite{jacquier2017randomized} studied its small-time and large-time behaviors of the implied volatility. The dynamic of the asset price $(S_t^{(\nu)})_{t \geq 0}$ and its stochastic volatility $(v_t^{\nu})_{t \geq 0}$ in the Stationary Heston model are given by
\begin{equation} \label{RH:model_stationary_heston}
	\left\{
	\begin{aligned}
		\frac{dS_t^{(\nu)}}{S_t^{(\nu)}} & = (r-q) dt + \sqrt{v_t^{\nu}} \big( \rho d \widetilde W_t + \sqrt{1 - \rho^2} d W_t \big) \\
		dv_t^{\nu}                       & = \kappa (\theta - v_t^{\nu}) dt + \xi \sqrt{v_t^{\nu}} d \widetilde W_t
	\end{aligned}
	\right.
\end{equation}
where $v_0^{\nu} \sim \Law (\nu) \sim \Gamma(\alpha, \beta)$ with $\beta = 2 \kappa / \xi^2$, $\alpha = \theta \beta$. $S_0^{(\nu)}$, $r$ and $q$ are the same parameters as those defined in \eqref{RH:model_standard_heston} and the parameters $\rho$, $\theta$, $\kappa$, $\theta$ and $\xi$ can be described as in the Standard Heston model.

\section{Pricing of European Options and Calibration} \label{RH:section:pricing_calibrationEU}

In this section, we first calibrate both Stationary and Standard Heston models and then compare their short-term behaviors of their resulting implied volatility surfaces. For that purpose we relied on a dataset of options price on the \textsc{Euro Stoxx} 50 observed the 26th of September 2019 (see Figure \ref{RH:fig:impliedvol_market}). This is why, as a preliminary step we briefly recall the well-known methodology for the evaluation of European Call and Put in the Standard Heston model. Based on that, we outline how to price these options in the Stationary Heston model. Then, we describe the methodology employed for the calibration of both models: the Stationary Heston model \eqref{RH:model_stationary_heston} and the Standard Heston model \eqref{RH:model_standard_heston} and then we discuss the obtained parameters and compare their short-term behaviors.

\subsection{European options}

The price of the European option with payoff $\varphi$ on the asset $S_T^{(\nu)}$, under the Stationary Heston model, exercisable at time $T$ is given by
\begin{equation}
	I_0 = \E \big[ \e^{-r T} \varphi ( S_T^{(\nu)} ) \big].
\end{equation}
After preconditioning by $v_0^{\nu}$, we have
\begin{equation} \label{RH:price_stationary_heston_eu}
	I_0 = \E \Big[ \E \big[ \e^{-r T} \varphi ( S_T^{(\nu)} ) \mid \sigma(v_0^{\nu}) \big] \Big] = \E \big[ f (v_0^{\nu}) \big]
\end{equation}
where $f(v)$ is the price of the European option in the Standard Heston model with deterministic initial conditions for the set of parameters $\lambda (v) = (s_0, r, q, \theta, \kappa, \xi, \rho, v)$.

\begin{example}[Call]
	If $\varphi$ is the payoff of a Call option then $f$ is simply the price given by Fourier transform in the Standard Heston model of the European Call Option. The price at time $0$, for a spot price $s_0$, of an European Call $C(\lambda(v), K, T)$ with expiry $T$ and strike $K$ under the Standard Heston model with parameters $\lambda(v) = (s_0, r, q, \theta, \kappa, \xi, \rho, v)$ is
	\begin{equation}\label{RH:price_call_heston}
		\begin{aligned}
			C(\lambda(v), K, T)
			 & = \E \big[ \e^{-r T} (S_T^{(v)} - K)_+ \big]                                                                    \\
			 & = \e^{-r T} \Big( \E \big[ S_T^{(v)} \1_{S_T^{(v)} \geq K} \big] - K \E \big[ \1_{S_T^{(v)} \geq K} \big] \Big) \\
			 & = s_0 \e^{-q T} P_1 \big( \lambda(v), K, T \big) - K \e^{-r T} P_2 \big(\lambda(v), K, T \big)
		\end{aligned}
	\end{equation}
	with $P_1 \big( \lambda(v), K, T \big)$ and $P_2 \big( \lambda(v), K, T \big)$ given by
	\begin{equation}
		\begin{aligned}
			P_1 \big( \lambda(v), K, T \big)
			 & = \frac{1}{2} + \frac{1}{\pi} \int_{0}^{+ \infty} \textrm{Re} \bigg( \frac{ \e^{-\textbf{i} u \log (K)} }{ iu } \frac{ \psi \big(\lambda(v), u - \textbf{i}, T \big) }{ s_0 \e^{ (r-q)T } } \bigg) du \\
			P_2 \big( \lambda(v), K, T \big)
			 & = \frac{1}{2} + \frac{1}{\pi} \int_{0}^{+ \infty} \textrm{Re} \bigg( \frac{ \e^{-\textbf{i} u \log (K)} }{ \textbf{i} u } \psi \big( \lambda(v), u, T \big) \bigg) du                                 \\
		\end{aligned}
	\end{equation}
	where $\textbf{i}$ is the imaginary unit s.t. $\textbf{i}^2 = -1$, $\psi \big( \lambda(v), u, T \big)$ is the characteristic function of the logarithm of the stock price process at time $T$. Several representations of the characteristic function exist, we choose to use the one proposed by \cite{SchoutensWim2004, gatheral2011volatility, littletrapheston}, which is numerically more stable. It reads
	\begin{equation}
		\begin{aligned}
			\psi \big( \lambda(v), u, T \big)
			 & = \E \big[ \e^{ \textbf{i} u \log(S_T^{(v)}) } \mid S_0^{(v)}, x \big]                                                            \\
			 & = \e^{ \textbf{i} u ( \log(s_0) + (r-q)T ) }                                                                                      \\
			 & \qquad \times \e^{ \theta \kappa \xi^{-2} \big( (\kappa - \rho \xi u \textbf{i} - d)T - 2 \log ( (1-g \e^{-dt}) / (1-g) ) \big) } \\
			 & \qquad \qquad \times \e^{ v^2 \xi^{-2} (\kappa - \rho \xi u \textbf{i} - d) (1 - \e^{-dt}) / (1-g\e^{-dt}) }                      \\
		\end{aligned}
	\end{equation}
	with
	\begin{equation}
		d = \sqrt{ ( \rho \xi u \textbf{i} - \kappa )^2 - \xi^2( - u \textbf{i} - u^2 ) } \quad \mbox{ and }\quad g = (\kappa - \rho \xi u \textbf{i} - d) / (\kappa - \rho \xi u \textbf{i} + d).
	\end{equation}
	Hence, in \eqref{RH:price_stationary_heston_eu}, $f(v)$ can be replaced by $C \big( \lambda(v), K, T \big)$, which yields
	\begin{equation}
		I_0 = \E \big[ \e^{-r T} ( S_T^{(\nu)} - K )_+ \big] = \E \Big[ C \big( \lambda(v_0^{\nu}), K, T \big) \Big].
	\end{equation}

\end{example}

Now, we come to the pricing of European options in the Stationary Heston model, using the expression of the density of $v_0^{\nu} \sim \Gamma (\alpha, \beta)$, \eqref{RH:price_stationary_heston_eu} reads
\begin{equation}
	I_0 = \E \big[ f (v_0^{\nu}) \big] = \int_0^{+ \infty} f (v) \frac{\beta^{\alpha}}{\Gamma(\alpha)} v^{\alpha -1} \e^{- \beta v} dv.
\end{equation}
Now, several approaches exists in order to approximate this integral on the positive real line.

\begin{itemize}[wide=0pt]
	\item {\em Quantization based quadrature formulas.} One could use a quantization-based cubature formula with an optimal quantizer of $v_0^{\nu}$ with the methodology detailed in Appendix \ref{RH:appendix:optquant}. Given that optimal quantizer of size $N$, $\widehat{v}_0^N$, we approximate $I_0$ by $\widehat{I}_0^N$
	      \begin{equation}
		      \widehat{I}_0^N = \E \big[ f (\widehat{v}_0^N) \big] = \sum_{i=1}^{N} f (v_{0,i}^N) \Prob \big( \widehat{v}_{0}^N = v_{0,i}^N \big).
	      \end{equation}

	      \begin{remarks}
		      In one dimension, the minimization problem, that consists in building an optimal quantizer, is invariant by linear transformation. Hence applying a linear transformation to an optimal quantizer preserves its optimality. For example, if we consider an optimal quantization $\widehat X^N$ of a standard normal distribution $\N (0,1)$ then $\mu + \sigma \widehat X^N$ is an optimal quantizer of a normal distribution $\N (\mu, \sigma^2)$ and the associated probabilities of each Voronoï centroid stay the same.

		      In our case, noticing that if we consider a Gamma random variable $X \sim \Gamma(\alpha, 1)$ then the rescaling of $X$ by $1/\beta$ yields $X/\beta \sim \Gamma(\alpha, \beta)$. Hence, for building the optimal quantizer $\widehat v_0^N$ of $v_0^{\nu}$, we can build an optimal quantizer of $X \sim \Gamma(\alpha, 1)$ and then rescale it by $1 / \beta$, yielding $\widehat v_0^N = \widehat{X}^N / \beta$. Our numerical tests showed that it is numerically more stable to use this approach.


		      In order to build the optimal quantizer, we use Lloyd's method detailed in Appendix \ref{RH:appendix:optquant} to $X \sim \Gamma(\alpha, 1)$ with the cumulative distribution function $F_{_X}( x ) = \Prob (X \leq x)$ and the partial first moment $K_{_X}( x ) = \E [X \1_{X \leq x}]$ given by
		      \begin{equation}
			      \begin{aligned}
				      \forall x > 0, \qquad       & F_{_X}( x ) = \frac{1}{\Gamma(\alpha)} \gamma(\alpha,x), \qquad & K_{_X}( x ) = \alpha F_{_X}( x ) - \frac{x^{\alpha} \e^{-x}}{\Gamma(\alpha)}, \\
				      \textrm{otherwise, } \qquad & F_{_X}( x ) = 0, \qquad                                         & K_{_X}( x ) = 0,                                                              \\
			      \end{aligned}
		      \end{equation}
		      where $\gamma(\alpha, x) = \int_{0}^{x} t^{\alpha -1} \e^{-t} dt$ is the lower gamma function. And the associated probabilities of the optimal quantizer $\widehat v_0^N$ are given by \eqref{RH:probaOptimalQuantizer}
		      \begin{equation}\label{RH:probainitialmeasure}
			      \Prob \big( \widehat v_0^N = v_{0,i}^N \big) = \Prob \big( \widehat{X}^N = x_i^N \big) = F_{_X} \big( x_{i + 1/2}^N \big) - F_{_X} \big( x_{i - 1/2}^N \big)
		      \end{equation}
		      where $\forall i \in \llbracket 2, N \rrbracket, x_{i-1/2}^N = \frac{x_{i-1}^N + x_i^N}{2}$ and $x_{1/2}^N = 0$ and $x_{N+1/2}^N = + \infty$.
	      \end{remarks}

	\item {\em Quadrature formula from Laguerre polynomials.} One could also use an algorithm based on fixed point quadratures for the numerical integration. Indeed, noticing that the density we are integrating against is a gamma density which is exactly the Laguerre weighting function (up to a rescaling). Then, $I_0$ rewrites
	      \begin{equation}
		      I_0
		      = \int_0^{+ \infty} f (v) \frac{\beta^{\alpha}}{\Gamma(\alpha)} v^{\alpha -1} \e^{- \beta v} dv
		      = \frac{\beta^{\alpha}}{\Gamma(\alpha)} \int_0^{+ \infty} f (v) \omega(v) dv
	      \end{equation}
	      where $\omega (v) = v^{\alpha-1} \e^{-\beta v}$ is the Laguerre weighting function. Then, for a fixed integer $n \geq 1$\footnote{In practice, we choose $n=20$. This number of points allows us to reach a high precision while keeping the computation time under control.}, $I_0$ is approximated by
	      \begin{equation}
		      \widetilde I_0^n = \frac{\beta^{\alpha}}{\Gamma(\alpha)} \sum_{i=1}^n \omega_i f(v_i)
	      \end{equation}
	      where the $\omega_i$'s are the Laguerre weights and the $v_i$'s are the associated Laguerre nodes.
\end{itemize}

\subsection{Calibration}
%
%
%
%
%
%
%
%

\begin{figure}
	\centering
	\includegraphics[width=1.\textwidth]{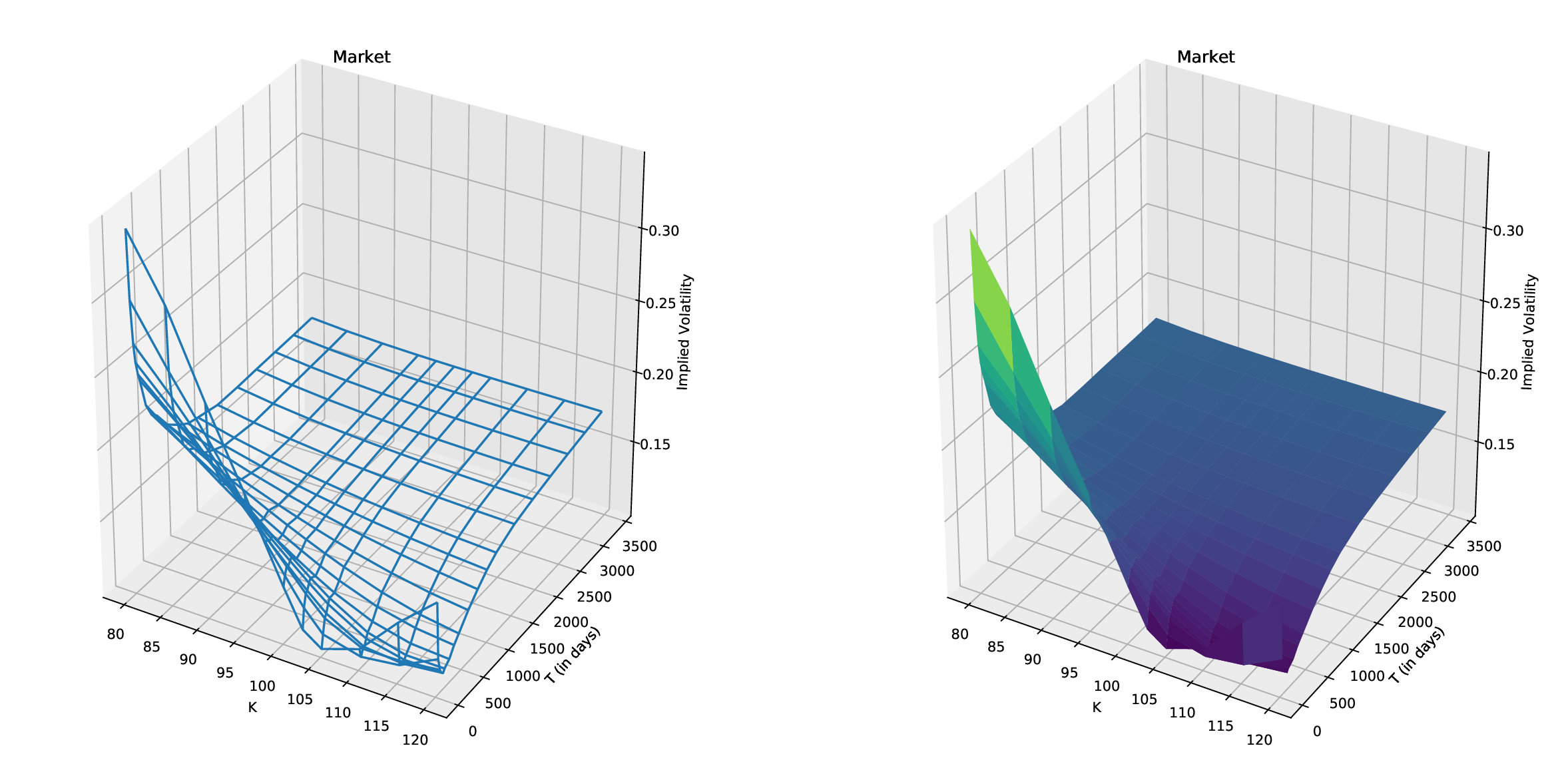}
	\caption[Implied volatility surface of the \textsc{Euro Stoxx} 50 as of the 26th of September 2019.]{\textit{Implied volatility surface of the \textsc{Euro Stoxx} 50 as of the 26th of September 2019. ($S_0 = 3541$, $r=-0.0032$ and $q=0.00225$) The expiries $T$ are given in days and the strikes $K$ in percentage of the spot.}}
	\label{RH:fig:impliedvol_market}
\end{figure}

Now that we are able to compute the price of European options, we define the problem of minimization we wish to optimize in order to calibrate our models parameters. Let $\mathcal{P}_{_{\sc SH}}$ be the set of parameters of the Stationary Heston model that needs to be calibrated, defined by
\begin{equation}
	\mathcal{P}_{_{\sc SH}} = \big\{ (\theta, \kappa, \xi, \rho) \in \R_+ \times \R_+ \times \R_+ \times [-1,1] \big\}
\end{equation}
and let $\mathcal{P}_{_{\sc H}}$ be the set of parameters of the Standard Heston model that needs to be calibrated, defined by
\begin{equation}
	\mathcal{P}_{_{\sc H}} = \big\{ (x, \theta, \kappa, \xi, \rho) \in \R_+ \times \R_+ \times \R_+ \times \R_+ \times [-1,1] \big\}.
\end{equation}

The others parameters are directly inferred from the market: we get $S_0 = 3541$, $r=-0.0032$ and $q=0.00225$.
In our case, we calibrate to option prices all having the same maturity. The problem can be formulated as follows: we search for the set of parameters $\phi^{\star} \in \mathcal{P}$ that minimizes the relative error between the implied volatility observed on the market and the implied volatility produced by the model for the given set of parameters, such that $\mathcal{P} = \mathcal{P}_{_{\sc SH}}$ for the Stationary Heston model and $\mathcal{P} = \mathcal{P}_{_{\sc H}}$ for the Standard Heston model. There is no need to calibrate the parameters $s_0$, $r$ and $q$ since they are directly observable in the market.

Being interested in the short-term behaviors of the models, it is natural to calibrate both models based on options prices at a small expiry. Once the optimization procedures have been performed, we compare their performances for small expiries. For that, we calibrate using only the data on the volatility surface in Figure \ref{RH:fig:impliedvol_market} with expiry 50 days ($T = 50/365$) and then we compare both models to the market implied volatility at expiry 22 days which is the smallest available in the data set.

\begin{remark}
	The calibration is performed in C++ on a laptop with a 2,4 GHz 8-Core Intel Core i9 CPU using the randomized version of the simplex algorithm of \cite{nelder1965simplex} proposed in the C++ library GSL. This algorithm is a derivative-free optimization method. It uses only the value of the function at each evaluation point. The computation time for calibrating the Standard Heston model is around $20$s and a bit more than a minute for the Stationary model. However, these computation times need to be considered carefully because the calibration time highly depends on the initial condition we choose for the minimizer and on the implementation of the Call pricer in the Standard Heston model.
\end{remark}

\subsubsection{Optimization without penalization}

We want to find the set of parameter $\phi^{\star}$ that minimizes the relative error between the volatilities observed in the market and the ones generated by the model, hence leading to the following minimization problem
\begin{equation}
	\min_{ \phi \in \mathcal{P} } \sum_{K} \bigg( \frac{ \sigma_{\textsc{iv}}^{Market}(K,T) - \sigma_{\textsc{iv}}^{Model}(\phi, K,T) }{ \sigma_{\textsc{iv}}^{Market}(K,T) } \bigg)^2
\end{equation}
where $T$ is the expiry of the chosen options chosen a priori and $K$ are their strikes. $\sigma_{\textsc{iv}}^{Market}(K,T)$ is the Mark-to-Market implied volatility taken from the observed implied volatility surface and the implied volatility $\sigma_{\textsc{iv}}^{Model}(\phi, K,T)$ is the Black-Scholes volatility $\sigma$ that matches the European Call price in this model to the price given by the Standard or Stationary Heston model with the set of parameters $\phi$.

\medskip
In all the following figures, the strike $K$ is given in percentage of the spot $S_0$.
\begin{figure}[H]
	\centering
	\includegraphics[width=1.\textwidth]{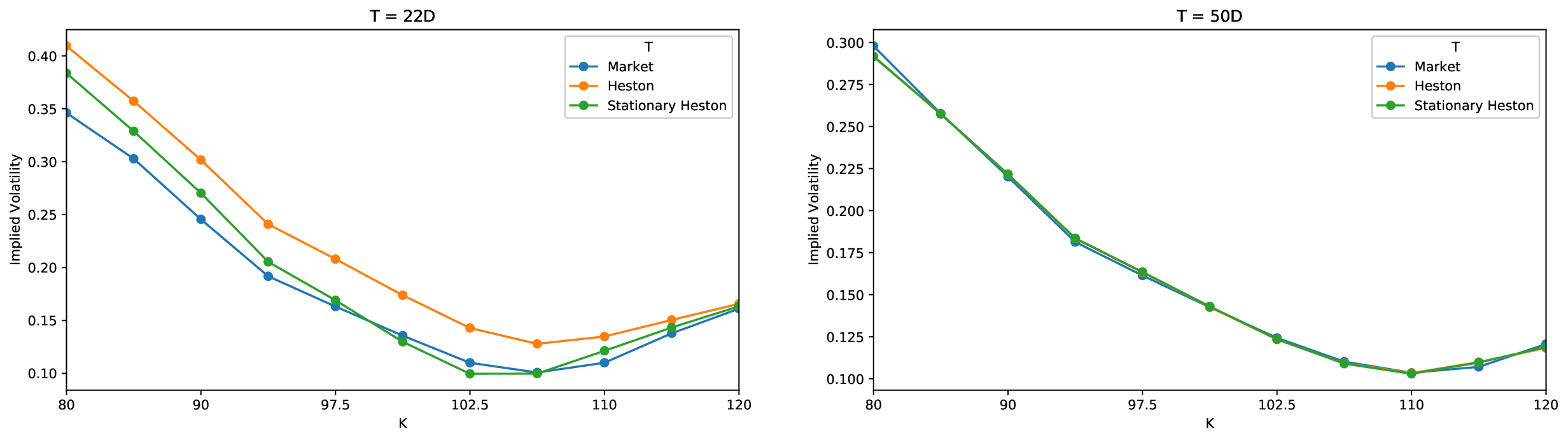}
	\caption[Implied volatilities for $22$ and $50$ days expiry options after calibration without penalization.]{\textit{Implied volatilities for $22$ (left) and $50$ (right) days expiry options after calibration at $50$ days without penalization.}}
	\label{RH:fig:impliedvol_22D_50D_nopena}
\end{figure}
It is clear in Figure \ref{RH:fig:impliedvol_22D_50D_nopena} (right) that both models fit really well to the market data and more precisely, the Stationary model succeeds to calibrate with the same precision as the Standard one with one less parameter. Moreover, one notices that even for 22 days maturity options, the Standard Heston model tends to over-estimate the implied volatility and fails to produce the right smile whereas the Stationary Heston model is closer to the market observations.

Now, we extrapolate the implied volatility surfaces, given by the two models, for even smaller maturities (7 and 14 days) in order to analyze the behavior of each model for short-term expiries.
\begin{figure}[H]
	\centering
	\includegraphics[width=1.\textwidth]{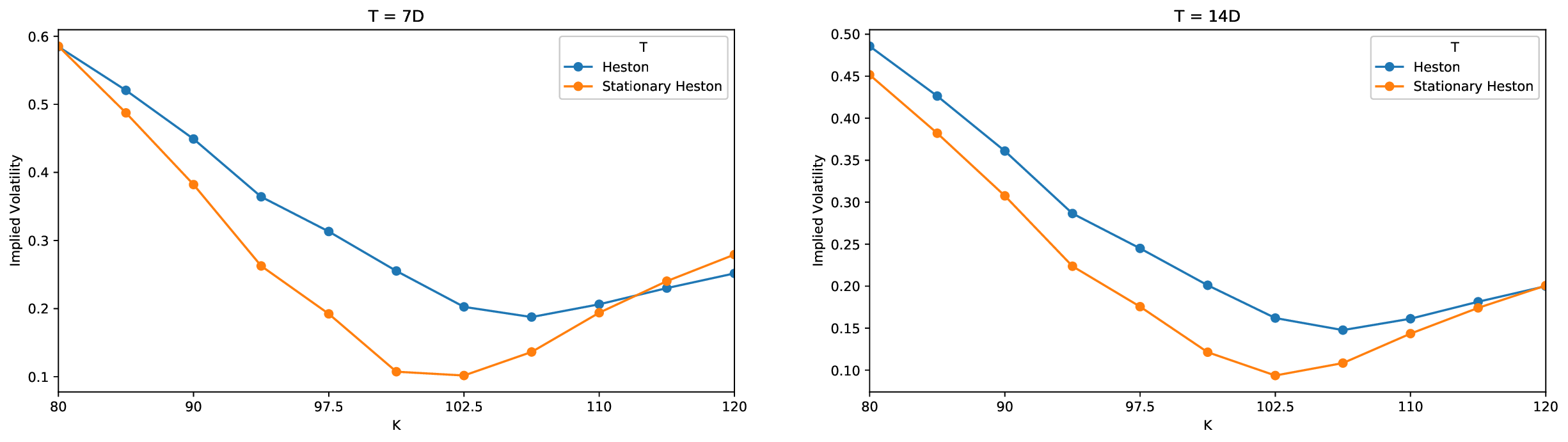}
	\caption[Implied volatilities for $7$ and $14$ days expiry options after calibration without penalization.]{\textit{Implied volatilities for $7$ (left) and $14$ (right) days expiry options after calibration at $50$ days without penalization.}}
	\label{RH:fig:impliedvol_7D_14D_nopena}
\end{figure}
It is clear in Figure \ref{RH:fig:impliedvol_7D_14D_nopena} that the Standard Heston model fails at producing the desired smile for very small maturities when the Stationary model meets no difficulty to generate it. The next graphics, Figure \ref{RH:fig:impliedvol_all_dates_nopena} reproduces the term-structure of the implied volatility in function of $T$ both models.
\begin{figure}[H]
	\centering
	\includegraphics[width=1.\textwidth]{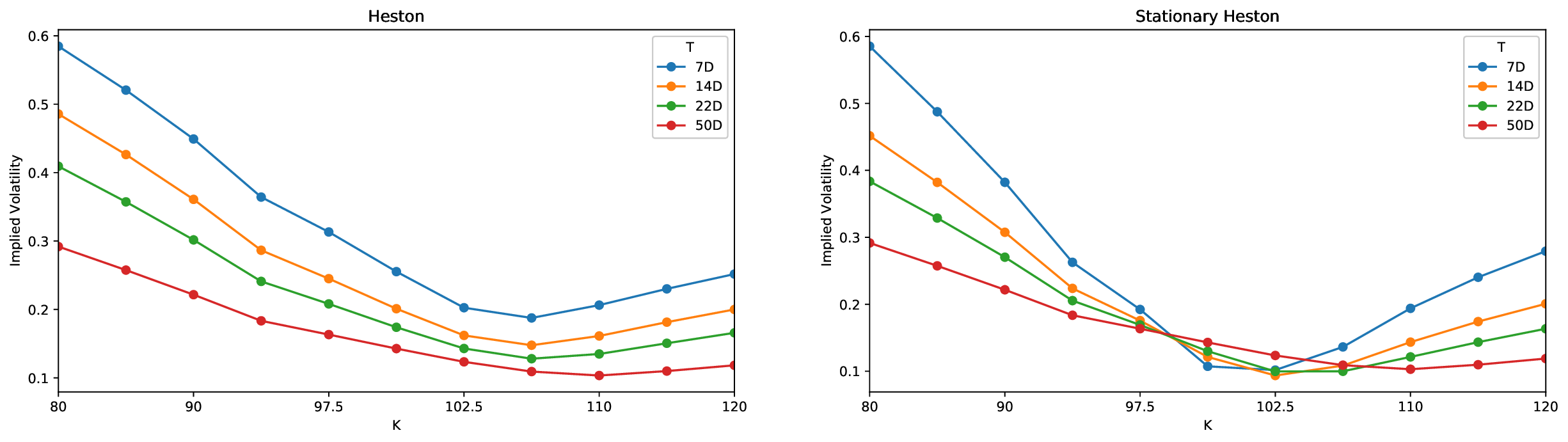}
	\caption[Term-structure of the volatility in function of $T$ and $K$ of both Heston models (stationary and standard) after calibration without penalization.]{\textit{Term-structure of the volatility in function of $T$ and $K$ of both models (left: Standard Heston and right: Stationary Heston) after calibration at $50$ days without penalization.}}
	\label{RH:fig:impliedvol_all_dates_nopena}
\end{figure}

Now, we investigate how these models behave for longer maturities. Do they succeed in preserving the general shape of the market volatility surface or are they only correctly fitting the maturity on which we calibrated them?

Figure \ref{RH:fig:impliedvol_rel_error_nopena} represents the relative error between the implied volatility given by the market and the one given by the models calibrated models at 50 days. Clearly, one notices that the Standard Heston model only fits at this expiry. Indeed, when looking at the expiry $22$ days or for long-term maturities, the relative error explodes. The term-structure of the implied volatility surface of the market is not preserved when using the Standard Heston model. However, the Stationary Heston model does fit well at both short and long term expiries. The Stationary model produces a steep smile for very short maturities and flattens correctly to the appropriate mean for long expiries.

\begin{table}[H]
	\centering
	\begin{tabular}{l||ccccc}
		\toprule
		$\phi^{\star}$    & $\rho$  & $v_0$      & $\theta$  & $\kappa$ & $\xi$   \\ \midrule \midrule
		Standard Heston   & $-0.74$ & $0.152584$ & $0.01487$ & $80.05$  & $5.22$  \\ \midrule
		Stationary Heston & $-0.75$ &            & $0.02744$ & $593.46$ & $36.80$ \\ \bottomrule
	\end{tabular}
	\caption[Parameters obtained for both models after calibration without penalization.]{\textit{Parameters obtained for both models after calibration without penalization for options with maturity $50$ days ($S_0 = 3541$, $r=-0.0032$ and $q=0.00225$).}}
	\label{RH:tab:params_nopena}
\end{table}

However, looking closely at the parameters obtained after calibration (which are summarized in Table \ref{RH:tab:params_nopena}), one notices that both sets of calibrated parameters are far from satisfying the Feller condition. And we have to keep in mind that the calibration procedure is performed in order to price path-dependent or American style derivatives using Monte-Carlo simulation or alternative numerical methods, as developed in the next Section. Hence, the Feller condition has to be satisfied, this is the reason why we add a constraint to the minimization problem in order to penalize the sets of parameters not satisfying the condition.

\begin{figure}[H]
	\centering
	\includegraphics[width=1.\textwidth]{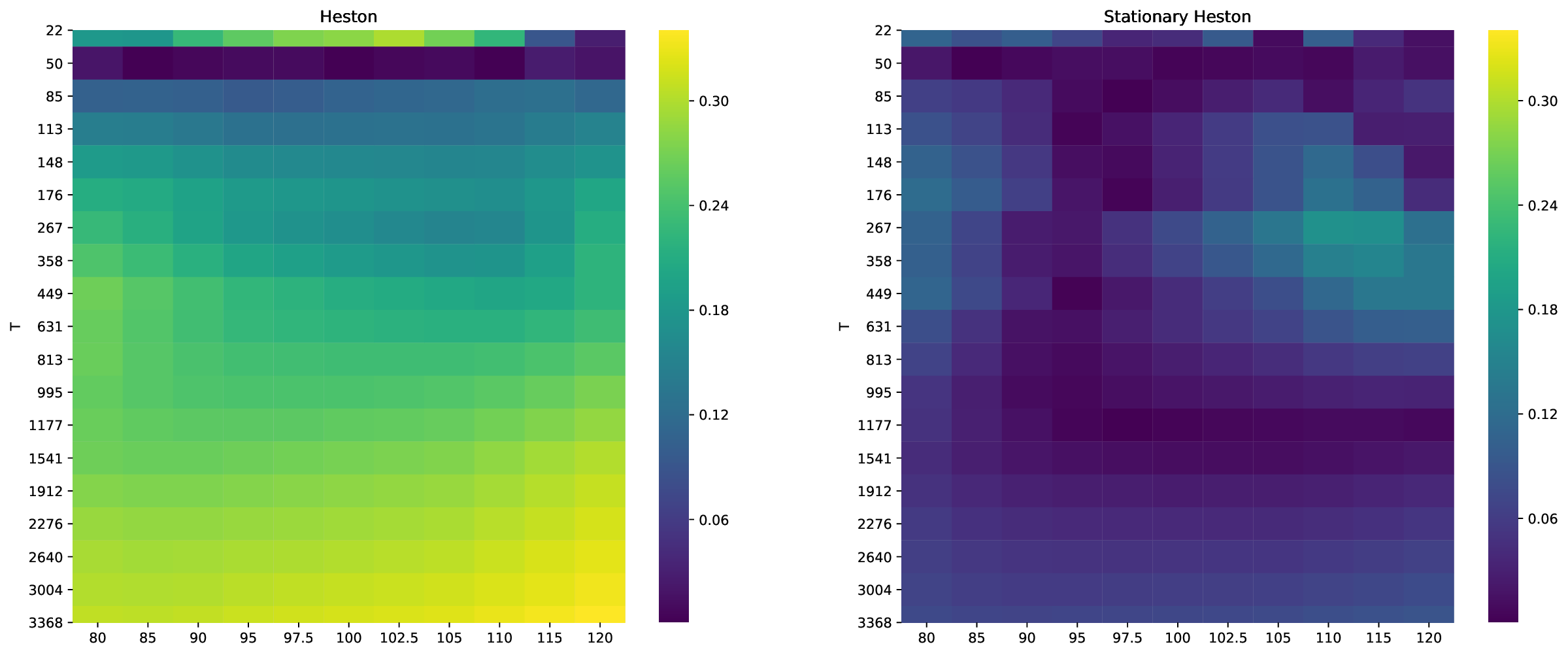}
	\caption[Relative error between market and models implied volatility after calibration without penalization.]{\textit{$(K,T) \longrightarrow \frac{\vert \sigma_{\textsc{iv}}^{Market}(K,T) - \sigma_{\textsc{iv}}^{Model}(\phi^{\star},K,T) \vert}{ \sigma_{\textsc{iv}}^{Market}(K,T) }$ for both models after calibration at $50$ days without penalization. The expiries $T$ are given in days and the strikes $K$ are in percentage of the spot. (left: Standard Heston and right: Stationary Heston)}.}
	\label{RH:fig:impliedvol_rel_error_nopena}
\end{figure}

\subsubsection{Optimization with penalization using the Feller condition}

The minimization problem becomes
\begin{equation}
	\min_{ \phi \in \mathcal{P} } \sum_{K} \bigg( \frac{ \sigma_{\textsc{iv}}^{Market}(K,T) - \sigma_{\textsc{iv}}^{Model}(\phi, K,T) }{ \sigma_{\textsc{iv}}^{Market}(K,T) } \bigg)^2 + \lambda \max( \xi^2 - 2 \kappa \theta, 0)
\end{equation}
where $\lambda$ is the penalization factor to be adjusted during the procedure. The obtained parameters after calibration are summarized in Table \ref{RH:tab:params_withpena}. The Feller condition is still not fulfilled for both models but it is not far from being satisfied. We choose $\lambda = 0.01$ which seems to be right the compromise in order to avoid underfitting the model because of the constraint.


\begin{table}[H]
	\centering
	\begin{tabular}{l||ccccc}
		\toprule
		$\phi^{\star}$    & $\rho$  & $v_0$    & $\theta$  & $\kappa$ & $\xi$  \\ \midrule \midrule
		Standard Heston   & $-0.83$ & $0.0045$ & $0.17023$ & $2.19$   & $1.04$ \\ \midrule
		Stationary Heston & $-0.99$ &          & $0.02691$ & $19.28$  & $1.15$ \\ \bottomrule
	\end{tabular}
	\caption[Parameters obtained for both models after calibration with penalization.]{\textit{Parameters obtained for both models after calibration with penalization ($\lambda = 0.01$) for options with maturity $50$ days ($S_0 = 3541$, $r=-0.0032$ and $q=0.00225$).}}
	\label{RH:tab:params_withpena}
\end{table}

Figure \ref{RH:fig:impliedvol_22D_50D_withpena} displays the resulting implied volatility curves at $50$ days and $22$ days for both calibrated models and observed in the market with calibration at $50$ days. Adding a penalization term deteriorates the calibration results compared to the non-penalized case (see Figure \ref{RH:fig:impliedvol_22D_50D_nopena} (right)) but the results are still acceptable.

\begin{figure}[H]
	\centering
	\includegraphics[width=1.\textwidth]{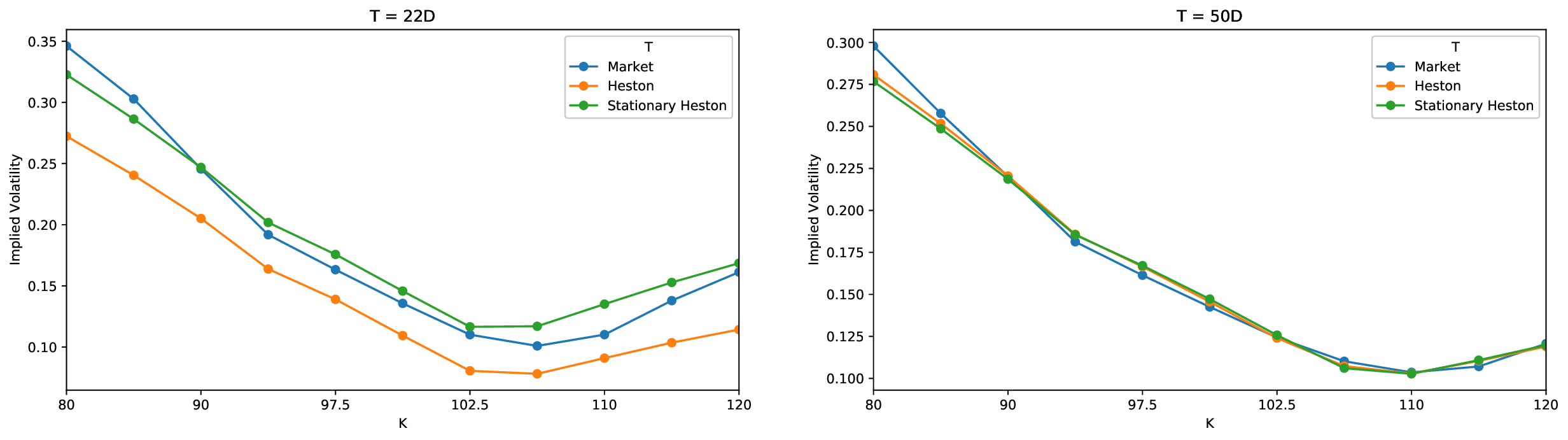}
	\caption[Implied volatilities for $22$ and $50$ days expiry options after calibration with penalization.]{\textit{Implied volatilities for $22$ (left) and $50$ (right) days expiry options after calibration at $50$ days with penalization.}}
	\label{RH:fig:impliedvol_22D_50D_withpena}
\end{figure}

Now, again, we extrapolate the implied volatility of both models for very short term maturities in Figure \ref{RH:fig:impliedvol_7D_14D_withpena}. The Stationary Heston model produces the desired smile, however the Standard Heston model fails to produce prices sensibly different than $0$ for strikes higher than $105$ with this set of parameters, this is why there is no values in implied volatility curves.

\begin{figure}[H]
	\centering
	\includegraphics[width=1.\textwidth]{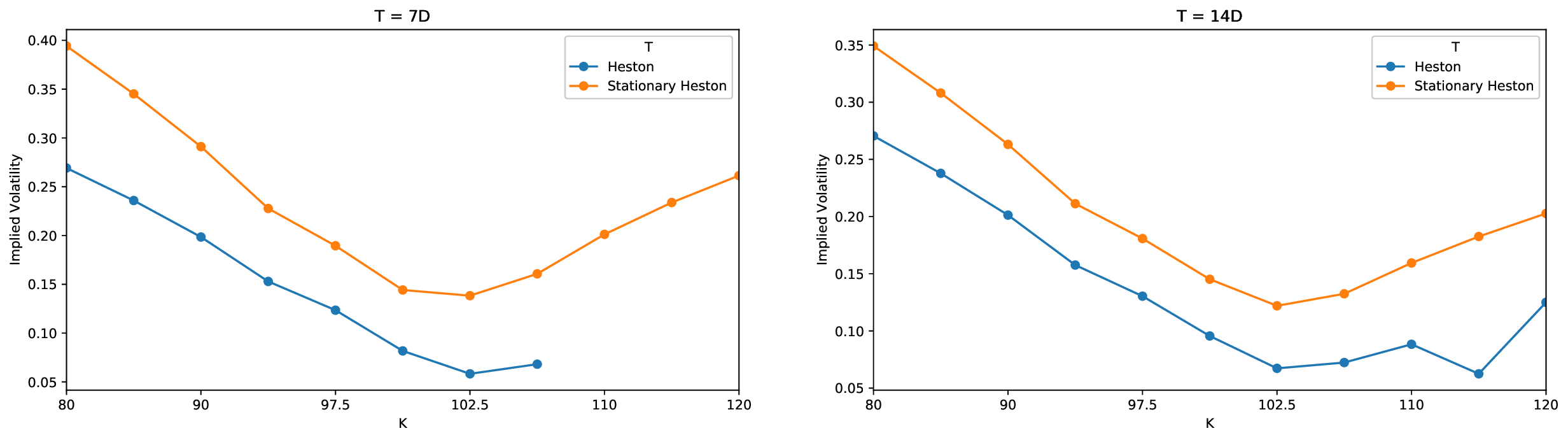}
	\caption[Implied volatilities for $7$ and $14$ days expiry options after calibration with penalization.]{\textit{Implied volatilities for $7$ (left) and $14$ (right) days expiry options after calibration at $50$ days with penalization.}}
	\label{RH:fig:impliedvol_7D_14D_withpena}
\end{figure}

Figure \ref{RH:fig:impliedvol_rel_error_withpena} represents, as in the non-penalized case, the relative error between the implied volatility given by the market and the one given by the models calibrated models at 50 days using a penalization. The Standard Heston model completely fails to preserve the term-structure while being calibrated at $50$ days. In comparison, the Stationary Heston behaves much better and the relative error does not explodes for long-term expiries, meaning that the long run average price variance is well caught.

\begin{figure}[H]
	\centering
	\includegraphics[width=1.\textwidth]{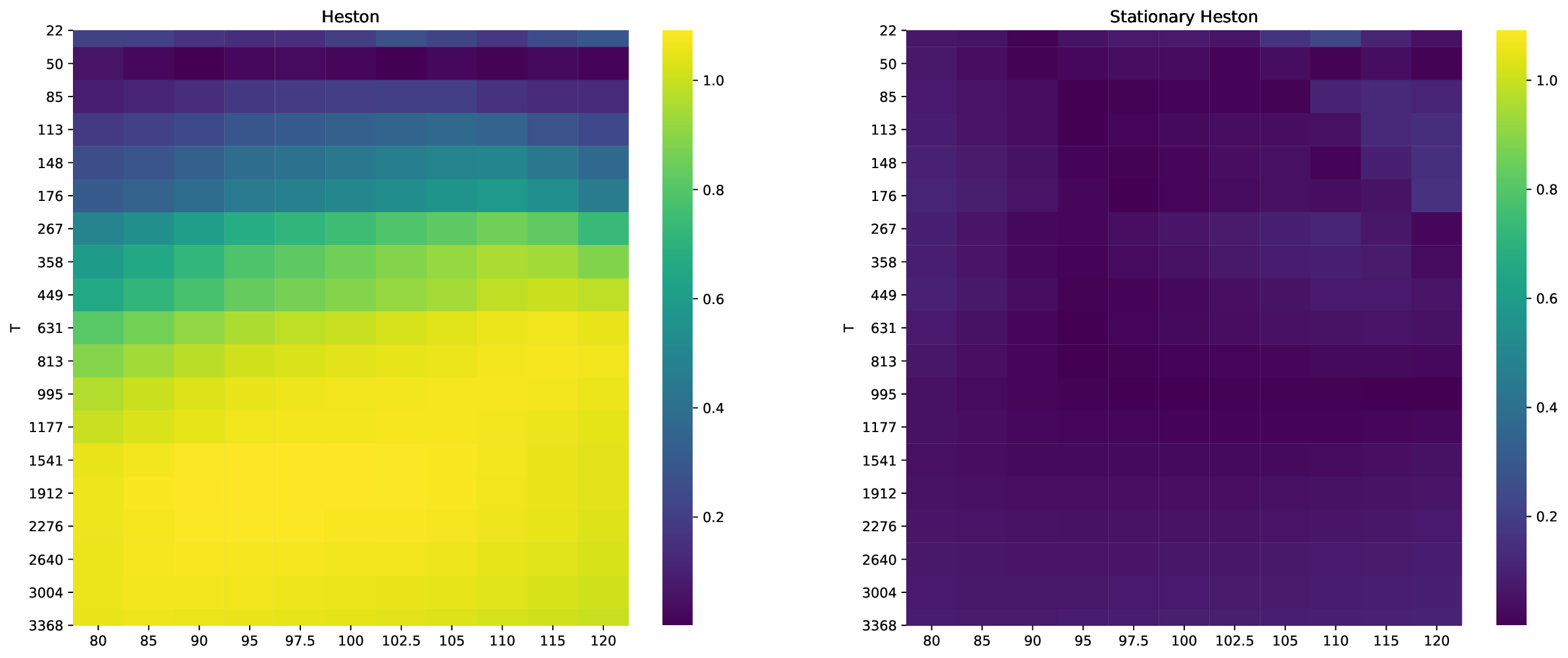}
	\caption[Relative error between market and models implied volatility after calibration with penalization.]{\textit{$(K,T) \longrightarrow \frac{\vert \sigma_{\textsc{iv}}^{Market}(K,T) - \sigma_{\textsc{iv}}^{Model}(\phi^{\star},K,T) \vert}{ \sigma_{\textsc{iv}}^{Market}(K,T) }$ for both models after calibration at $50$ days with penalization. The expiries $T$ are given in days and the strikes $K$ are in percentage of the spot. (left: Standard Heston and right: Stationary Heston)}.}
	\label{RH:fig:impliedvol_rel_error_withpena}
\end{figure}

\section{Toward the pricing of Exotic Options} \label{RH:section:exotic_options}

In this Section, we evaluate first Bermudan options and then Barrier options under the Stationary Heston model. For both products, the pricing rely on a \textit{Backward Dynamic Programming Principle}. The numerical solution we propose is based on a two-dimensional product recursive quantization scheme. We extend the methodology previously developed by \cite{abbas2018product,callegaro2017american,callegaro2017pricing}, where they considered an Euler-Maruyama scheme for both components. In this paper, we consider a hybrid scheme made up with an Euler-Maruyama scheme for the $\log$-stock price dynamics and a Milstein scheme for the (boosted) volatility process. Finally, we apply the backward algorithm that corresponds to the financial product we are dealing with (the \textit{ Quantized Backward Dynamic Programming Principle} for Bermudan Options, see \cite{bally2003quantization,printems2005quantization,pages2018numerical} and the algorithm by \cite{sagna2010pricing, pages2018numerical} for Barrier Options based on the conditional law of the Brownian motion).

\subsection{Discretization scheme of a stochastic volatility model}

We first present the time discretization schemes we use for the asset-volatility couple $(S_t^{(\nu)}, v_t^{\nu})_{t \in [0,T]}$. For the volatility, we choose a Milstein on a \textit{boosted} version of the process in order to preserve the positivity of the volatility and we select an Euler-Maruyama scheme for the $\log$ of the asset.

\paragraph{The boosted volatility.}
Based on the discussion in Appendix \ref{RH:appendix:discussionschememilstein}, we will work with the following \textit{boosted} volatility process: $Y_t = \e^{\kappa t} v_t^{\nu}, t \in [0,T]$ for some $\kappa > 0$, whose diffusion is given by
\begin{equation}
	d Y_t = \e^{\kappa t} \kappa \theta dt + \xi \e^{\kappa t / 2} \sqrt{Y_t} d \widetilde W_t.
\end{equation}

The Milstein discretization scheme of $Y_t$ is given by
\begin{equation}\label{RH:discretized_boostedvol}
	\widebar{Y}_{t_{k+1}} = \mathcal{M}_{\widetilde b, \widetilde \sigma} \big( t_k, \widebar{Y}_{t_k}, Z_{k+1}^2 \big)
\end{equation}
with $t_k = \frac{T k}{n}$ and $\widetilde b$ and $\widetilde \sigma$ are given by
\begin{equation}\label{RH:drift_and_vol_boostedvol}
	\widetilde b(t,x) = \e^{\kappa t} \kappa \theta, \qquad \widetilde \sigma(t,x) = \xi \sqrt{x} \e^{\kappa t / 2} \quad \textrm{ and } \quad \widetilde \sigma_{x}^{\prime} (t,x) = \frac{\xi \e^{\kappa t / 2}}{2 \sqrt{x}}
\end{equation}
and $\mathcal{M}_{\widetilde b,\widetilde \sigma} (t,x,z)$ defined by
\begin{equation}\label{RH:MilsteinScheme}
	\begin{aligned}
		\mathcal{M}_{\widetilde b,\widetilde \sigma} (t,x,z)
		 & = x - \frac{ \widetilde \sigma(t,x)}{2 \widetilde \sigma_{x}^{\prime} (t,x)} + h \bigg(\widetilde b(t, x) - \frac{ (\widetilde \sigma \widetilde \sigma_{x}^{\prime}) (t,x)}{2}\bigg) + \frac{ (\widetilde \sigma \widetilde \sigma_{x}^{\prime}) (t,x) h}{2} \bigg( z + \frac{1}{\sqrt{h} \widetilde \sigma_{x}^{\prime} (t,x)} \bigg)^2.
	\end{aligned}
\end{equation}
We made this choice of scheme because, under the Feller condition, the positivity of $\mathcal{M}_{\widetilde b,\widetilde \sigma}$ is ensured, since
\begin{equation}\label{RH:MilsteinScheme_with_specificationmodel}
	\begin{aligned}
		\mathcal{M}_{\widetilde b,\widetilde \sigma} (t,x,z)
		 & = h \e^{\kappa t} \Big( \kappa \theta - \frac{\xi^2}{4} \Big) + h \frac{\xi^2 \e^{\kappa t}}{4} \bigg( z + \frac{2 \sqrt{x}}{\sqrt{h} \xi \e^{ \kappa t / 2} } \bigg)^2
	\end{aligned}
\end{equation}
and
\begin{equation*}
	\xi^2 \leq 2 \kappa \theta \leq 4 \kappa \theta.
\end{equation*}

Other schemes could have been used, see \cite{alfonsi2005discretization} for an extensive review of the existing schemes for the discretization of the CIR model, but in our case we needed one allowing us to use the fast recursive quantization, i.e., where we can express explicitly and easily the cumulative distribution function and the first partial moment of the scheme, which is the case of the Milstein scheme (we give more details in SubSection \ref{RH:subsection:hybridproductrecursivequantization}).

Hence, as our time-discretized scheme is well defined because its positivity is ensured if the Feller condition is satisfied, we can start to think of the time-discretization of our process $( S_{t_k}^{(\nu)} )_{k \in \llbracket 0, n \rrbracket}$.

\paragraph{The $\log$-asset.}

For the asset, the standard approach is to consider the process which is the logarithm of the asset $X_t = \log (S_t)$. Applying Itô's formula, the dynamics of $X_t$ is given by
\begin{equation}
	d X_t = \Big( r - q - \frac{v_t}{2} \Big) dt + \sqrt{v_t} d W_t.
\end{equation}
Now, using a standard Euler-Maruyama scheme for the discretization of $X_t$, we have
\begin{equation}\label{RH:discretized_couple}
	\left\{
	\begin{aligned}
		\widebar{X}_{t_{k+1}} & = \mathcal{E}_{b, \sigma} \big( t_k, \widebar{X}_{t_k}, \widebar{Y}_{t_k}, Z_{k+1}^1 \big)    \\
		\widebar{Y}_{t_{k+1}} & = \mathcal{M}_{\widetilde b, \widetilde \sigma} \big( t_k, \widebar{Y}_{t_k}, Z_{k+1}^2 \big)
	\end{aligned}
	\right.
\end{equation}
where $Z_{k+1}^1 \sim \N(0,1)$, $Z_{k+1}^2 \sim \N(0,1)$, $\Corr (Z_{k+1}^1, Z_{k+1}^2) = \rho$ and
\begin{equation} \label{RH:eulerscheme}
	\mathcal{E}_{b, \sigma} (t,x,y,z) = x + b(t,x,y) h + \sigma (t,x,y) \sqrt{ h } \, z
\end{equation}
with
\begin{equation}\label{RH:drift_and_vol_logasset}
	b(t,x,y) = r - q - \frac{\e^{-\kappa t} y}{2} \qquad \textrm{ and } \qquad \sigma (t,x,y) = \e^{-\kappa t / 2} \sqrt{y}.
\end{equation}

\subsection{Hybrid Product Recursive Quantization} \label{RH:subsection:hybridproductrecursivequantization}

In this part, we describe the methodology used for the construction of the product recursive quantization tree of the couple log asset- boosted volatility in the Heston model.

In Figure \ref{RH:fig:quantif_recursive_explication}, as an example, we synthesise the main idea behind the recursive quantization of a diffusion $v_t$ which has been time-discretized with $F_0(t, x, z)$. We start at time $t_0=0$ with a quantizer $\widehat v_0$ taking values in the grid $\Gamma_{t_0} = \ \{ v_1^0, \dots, v_{10}^0 \}$ of size $10$, where each point is represented by a black bullet ($\bullet$) with probability $p_{i}^{0} = \Prob ( \widehat v_0 = v_i^0 )$ is represented by a bar. In the Stationary Heston model, $\widehat v_0$ is an optimal quantization of the Gamma distribution given by \eqref{RH:gammalaw} and \eqref{RH:params_gammalaw}. Then, starting from this grid, we simulate the process from time $t_0$ to time $t_1 = 5$ days with our chosen time-discretization scheme $F_0(t, x, z)$, yielding $\widetilde v_1 = F_0(t_0, \widehat v_0, Z_1)$, where $Z_1$ is a standardized Gaussian random variable. Each trajectory starts from point $v_i^0$ with probability $p_i^0$. And finally we project the obtained distribution at time $t_1$ onto a grid $\Gamma_{t_1} = \{ v_1^1, \dots, v_{10}^1 \}$ of cardinality $10$, represented by black triangles ($\blacktriangleup$) such that $\widehat v_1$ is an optimal quantizer of the discretized and simulated process starting from quantizer $\widehat v_0$ at time $t_0 = 0$.

\begin{remark}
	In practice, for low dimensions, we do not simulate trajectories. We use the information on the law of $\widetilde v_1$ conditionally of starting from $\widehat v_0$. The knowledge of the distribution allows us to use deterministic algorithms during the construction of the optimal quantizer of $\widetilde v_1$ that are a lot faster than algorithms based on simulation.
\end{remark}

\begin{figure}[h]
	\centering
	\includegraphics[width=1.\textwidth]{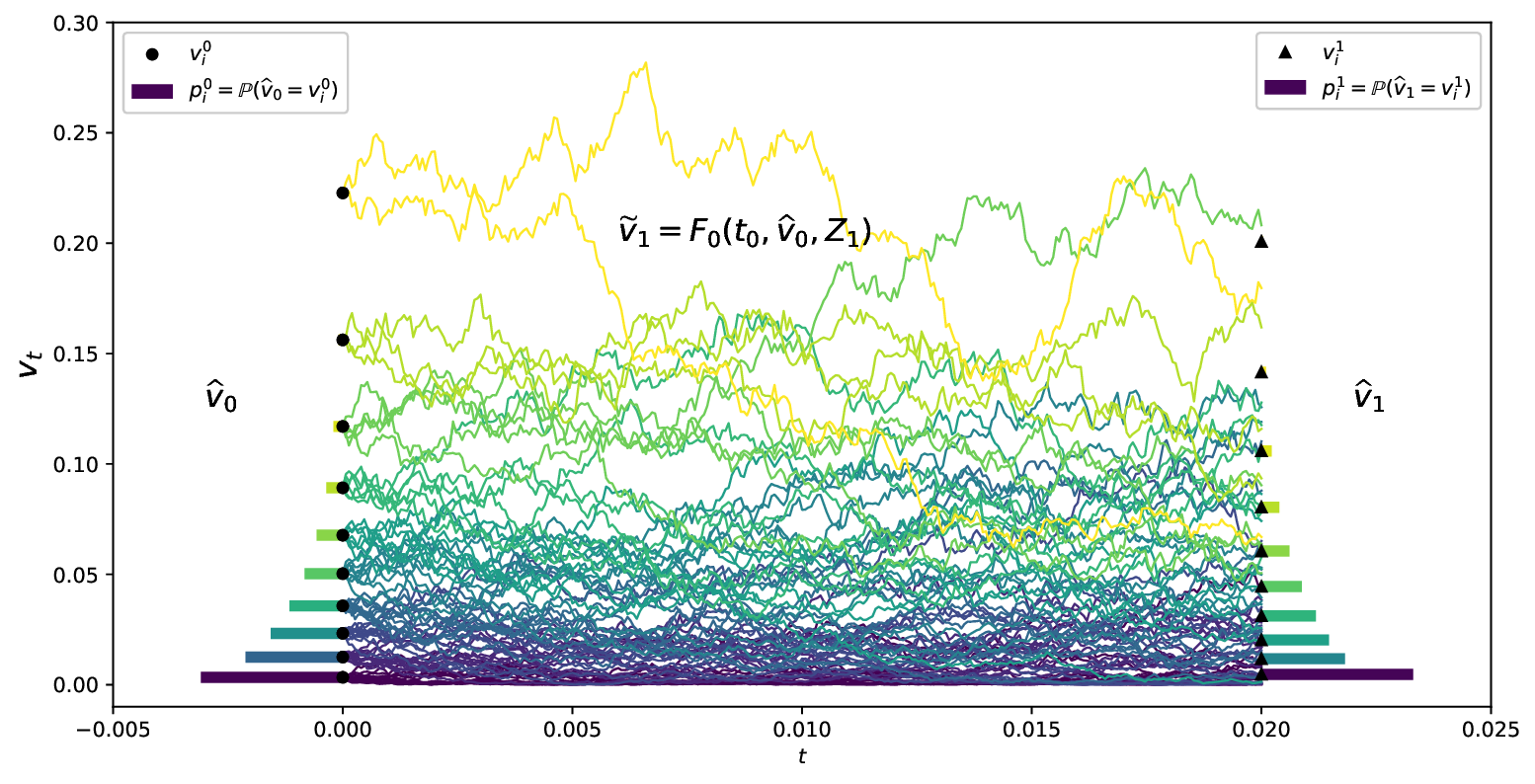}
	\caption[Example of recursive quantization of the volatility process in the Heston model for one time-step.]{\textit{Example of recursive quantization of the volatility process in the Heston model for one time-step.}}
	\label{RH:fig:quantif_recursive_explication}
\end{figure}

In our case, we consider the following stochastic volatility system
\begin{equation} \label{RH:systemvolsto}
	\left\{
	\begin{aligned}
		d X_t & = b (t, X_t, Y_t) dt + \sigma (t, X_t, Y_t) d W_t                        \\
		d Y_t & = \widetilde b (t, Y_t) dt + \widetilde \sigma (t, Y_t) d \widetilde W_t \\
	\end{aligned}
	\right.
\end{equation}
where $W_t$ and $\widetilde W_t$ are two correlated Brownian motions with correlation $\rho \in [-1, 1]$, $b$ and $\sigma$ are defined in \eqref{RH:drift_and_vol_logasset} and $\widetilde b$ and $\widetilde \sigma$ are defined in \eqref{RH:drift_and_vol_boostedvol}. Our aim is to build a quantization tree of the couple $(X_t, Y_t)$ at given dates $t_k, \, k=0, \dots, n$ based on a recursive product quantization scheme. The product recursive quantization of such diffusion system has already been studied by \cite{callegaro2017pricing} and \cite{rudd2017fast} in the case case where both processes are discretized using an Euler-Maruyama scheme.

One can notice that building the quantization tree $(\widehat{Y}_k)_{k \in \llbracket 0, n \rrbracket}$ approximating $( Y_t )_{t \in [0, T]}$ is a one dimensional problem as the diffusion of $Y_t$ is autonomous. Hence, based on our choice of discretization scheme, we will apply the fast recursive quantization (detailed above in Figure \ref{RH:fig:quantif_recursive_explication}) that was introduced in \cite{pages2015recursive} for one dimensional diffusion discretized by an Euler-Maruyama discretization scheme and then extended to higher order schemes, still in one dimension, by \cite{mcwalter2018recursive}. The minor difference with existing literature is that, in our problem, the initial condition $y_0$ is not deterministic.

Then, using the quantization tree of $(\widehat Y_k)_{k \in \llbracket 0, n \rrbracket}$ we will be able to build the tree $(\widehat{X}_k)_{k \in \llbracket 0, n \rrbracket}$ following ideas developed in \cite{abbas2018product,rudd2017fast,callegaro2017american,callegaro2017pricing}. Indeed, once the quantization tree of the volatility is built, we are in a one-dimensional setting and we are able to use fast deterministic algorithms.

\subsubsection{Quantizing the volatility (a one-dimensional case)}

Let $(Y_t)_{t \in [0,T]}$ be a stochastic process in $\R$ and solution to the stochastic differential equation
\begin{equation}
	d Y_t = \widetilde b (t, Y_t) dt + \widetilde \sigma (t, Y_t) d \widetilde W_t
\end{equation}
where $Y_0$ has the same law than the stationary measure $\nu$: $\Law (Y_0) = \nu$. In order to approximate our diffusion process, we choose a Milstein scheme for the time discretization, as defined in \ref{RH:MilsteinScheme} and we build recursively the Markovian quantization tree $(\widehat Y_{t_k})_{k \in \llbracket 0, n \rrbracket}$ where $\widehat Y_{t_{k+1}}$ is the Vorono\"i quantization of $\widetilde Y_{t_{k+1}}$ defined by
\begin{equation} \label{RH:recurQuantiVol}
	\widetilde Y_{t_{k+1}} = \mathcal{M}_{\widetilde b, \widetilde \sigma} \big( t_k, \widehat Y_{t_k}, Z_{k+1}^2 \big), \qquad \widehat Y_{t_{k+1}} = \Proj_{\Gamma_{N_{2,k+1}}^Y} \big( \widetilde{Y}_{t_{k+1}} \big)
\end{equation}
and the projection operator $\Proj_{\Gamma_{N_{2,k+1}}^Y}(\cdot)$ is defined in \eqref{RH:quantizOfX}, $\Gamma_{N_{2,k+1}}^Y = \big\{ y_1^{k+1}, \dots, y_{N_{2,k+1}}^k \big\}$ is the grid of the optimal quantizer of $\widetilde{Y}_{t_{k+1}}$ and $Z_{k+1}^2 \sim \N(0,1)$. In order alleviate the notations, we will denote $\widetilde Y_k$ and $\widehat Y_k$ in place of $\widetilde Y_{t_k}$ and $\widehat Y_{t_k}$.

\medskip
The first step consists in building $\widehat Y_0$, an optimal quantizer of size $N_{2,0}$ of $Y_0$. Noticing that $Y_0 = v_0^{\nu}$, we use the optimal quantizer we built for the pricing of European options. Then, we build recursively $(\widehat Y_k)_{k=1, \dots, n}$, where the $N_{2,k}$-tuple are defined by $y_{_{1:N_{2,k}}}^k = \big( y_1^k, \dots, y_{N_{2,k}}^k \big)$, by solving iteratively the minimization problem defined in the Appendix \ref{RH:appendix:optquant} in \eqref{RH:distortozero}, with the help of Lloyd's method I. Replacing $X$ by $\widetilde Y_{k+1}$ in \eqref{RH:distortozero} yields
\begin{equation}
	\begin{aligned}
		y_j^{k+1}
		 & = \frac{ \E \Big[ \widetilde Y_{k+1} \1_{ Y_{k+1} \, \in \, C_j \big( \Gamma_{N_{2,k+1}}^Y \big) } \Big] }{\Prob \Big( \widetilde Y_{k+1} \in C_j \big( \Gamma_{N_{2,k+1}}^Y \big) \Big) }                                                                                                                                                                                                                          \\
		 & = \frac{ \E \Big[ \mathcal{M}_{\widetilde b, \widetilde \sigma} \big( t_k, \widehat Y_k, Z_{k+1}^2 \big) \1_{ \mathcal{M}_{\widetilde b, \widetilde \sigma} \big( t_k, \widehat Y_k, Z_{k+1}^2 \big) \, \in \, C_j \big( \Gamma_{N_{2,k+1}}^Y \big) } \Big] }{ \Prob \Big( \mathcal{M}_{\widetilde b, \widetilde \sigma} \big( t_k, \widehat Y_k, Z_{k+1}^2 \big) \in C_j \big( \Gamma_{N_{2,k+1}}^Y \big) \Big) }.
	\end{aligned}
\end{equation}
Now, preconditioning by $\widehat Y_k$ in the numerator and the denominator and using $p_i^k = \Prob \big( \widehat Y_k = y_i^k \big)$, we have
\begin{equation}
	\begin{aligned}
		y_j^{k+1}
		 & = \frac{ \E \bigg[ \E \Big[ \mathcal{M}_{\widetilde b, \widetilde \sigma} \big( t_k, \widehat Y_k, Z_{k+1}^2 \big) \1_{ \mathcal{M}_{\widetilde b, \widetilde \sigma} \big( t_k, \widehat Y_k, Z_{k+1}^2 \big) \, \in \, C_j \big( \Gamma_{N_{2,k+1}}^Y \big) } \mid \widehat Y_k \Big] \bigg] }{ \E \bigg[ \Prob \Big( \mathcal{M}_{\widetilde b, \widetilde \sigma} \big( t_k, \widehat Y_k, Z_{k+1}^2 \big) \in C_j \big( \Gamma_{N_{2,k+1}}^Y \big) \mid \widehat Y_k \Big) \bigg] } \\
		 & = \frac{ \displaystyle \sum_{i=1}^{N_{2,k}} \E \Big[ \mathcal{M}_{\widetilde b, \widetilde \sigma} \big( t_k, y_i^k, Z_{k+1}^2 \big) \1_{ \mathcal{M}_{\widetilde b, \widetilde \sigma} \big( t_k, y_i^k, Z_{k+1}^2 \big) \, \in \, C_j \big( \Gamma_{N_{2,k+1}}^Y \big) } \Big] \, p_i^k}{ \displaystyle \sum_{i=1}^{N_{2,k}} \Prob \Big( \mathcal{M}_{\widetilde b, \widetilde \sigma} \big( t_k, y_i^k, Z_{k+1}^2 \big) \in C_j \big( \Gamma_{N_{2,k+1}}^Y \big) \Big) \, p_i^k}      \\
		 & = \frac{ \displaystyle \sum_{i=1}^{N_{2,k}} \Big( K_i^k \big( y_{j + 1/2}^{k+1} \big) - K_i^k \big( y_{j - 1/2}^{k+1} \big) \Big) \, p_i^k}{ \displaystyle \sum_{i=1}^{N_{2,k}} \Big( F_i^k \big( y_{j + 1/2}^{k+1} \big) - F_i^k \big( y_{j - 1/2}^{k+1} \big) \Big) \, p_i^k }                                                                                                                                                                                                         \\
	\end{aligned}
\end{equation}
where $C_j \big( \Gamma_{N_{2,k+1}}^Y \big) = \big( y_{j - 1/2}^{k+1}, y_{j + 1/2}^{k+1} \big]$ is defined in \eqref{RH:def_voronoi_cells}. $F_i^k$ and $K_i^k$ are the cumulative distribution function and the first partial moment function of $U_i^k \sim \mu_i^k + \kappa_i^k ( Z_{k+1}^1 + \lambda_i^k )^2$ respectively with
\begin{equation} \label{RH:mean_std_milstein}
	\begin{aligned}
		\kappa_j^k                   & = \frac{(\widetilde \sigma \widetilde \sigma_{x}^{\prime}) (t_k,y_j^k) h}{2}, \qquad \qquad \lambda_j^k = \frac{1}{\sqrt{h} \widetilde \sigma_{x}^{\prime} (t_k,y_j^k)},                                   \\
		\textrm{and } \qquad \mu_j^k & = y_j^k - \frac{\sigma(t_k,y_j^k)}{2 \widetilde \sigma_{x}^{\prime} ( t_k, y_j^k )} + h \bigg( \widetilde b(t_k, y_j^k) - \frac{(\widetilde \sigma \widetilde \sigma_{x}^{\prime}) (t_k,y_j^k)}{2} \bigg).
	\end{aligned}
\end{equation}
The functions $F_i^k$ and $K_i^k$ can explicitly be determined in terms of the density and the cumulative distribution function of the normal distribution.

\begin{lemme}
	Let $U = \mu + \kappa (Z + \lambda)^2$, with $\mu, \kappa, \lambda \in \R$, $\lambda \geq 0$, $\kappa > 0$ and $Z \sim \N (0, 1)$ then the cumulative distribution function $F_{_X}$ and the first partial moment $K_{_U}$ of $U$ are given by
	\begin{equation}
		\begin{aligned}
			F_{_U} ( x ) & = \big( F_{_Z} ( x_{_+} ) - F_{_Z} ( x_{_-} ) \big) \1_{x > \mu }                                                                                                                                       \\
			K_{_U} ( x ) & = \bigg( F_{_U} ( x ) \big( \mu + \kappa ( \lambda^2 + 1 ) \big) + \frac{ \kappa }{ \sqrt{2 \pi} } \Big( x_{_-} \e^{- \frac{x_{_+}^2}{2}} - x_{_+} \e^{- \frac{x_{_-}^2}{2} } \Big) \bigg) \1_{x > \mu} \\
		\end{aligned}
	\end{equation}
	where $x_{_+} = \sqrt{ \frac{x - \mu}{\kappa} } - \lambda$, $x_{_-} = - \sqrt{ \frac{x - \mu}{\kappa} } - \lambda $ and $F_{_Z}$ is the cumulative distribution function of $Z$.
\end{lemme}

Finally, we can apply the Lloyd algorithm defined in Appendix \ref{RH:LloydAlgo} with $F_{_X}$ and $K_{_X}$ defined by
\begin{equation}
	\begin{aligned}
		F_{_X} ( x ) = \sum_{i=1}^{N_{2,k}} p_i^k \, F_i^k ( x ) \qquad \textrm{ and } \qquad K_{_X} ( x ) = \sum_{i=1}^{N_{2,k}} p_i^k \, K_i^k ( x ).
	\end{aligned}
\end{equation}

In order to be able to build recursively the tree quantization $(\widehat Y_k)_{k = 0, \dots, n}$, we need to have access to the weights $p_i^k = \Prob \big( \widehat Y_k = y_i^k \big)$, which can be themselves computed recursively, as well as the conditional probabilities $p_{ij}^k = \Prob \big( \widehat Y_{k+1} = y_j^{k+1} \mid \widehat Y_k = y_i^k \big)$.

\begin{lemme}
	The conditional probabilities $p_{ij}^k$ are given by
	\begin{equation}\label{RH:condiProbavol}
		\begin{aligned}
			p_{ij}^k
			= & F_i^k \big( y_{j + 1/2}^{k+1} \big) - F_i^k \big( y_{j - 1/2}^{k+1} \big).
		\end{aligned}
	\end{equation}
	And the probabilities $p_j^{k+1}$ are given by
	\begin{equation}
		p_j^{k+1} = \sum_{i=1}^{N_{2,k}} p_i^{k} p_{ij}^k.
	\end{equation}
\end{lemme}

\begin{proof}
	The
	\begin{equation*}
		\begin{aligned}
			p_{ij}^k
			= & \Prob \big( \widehat Y_{k+1} = y_j^{k+1} \mid \widehat Y_k = y_i^k \big)                                                                                                    \\
			= & \Prob \Big( \mathcal{M}_{\widetilde b, \widetilde \sigma} \big( t_k, \widehat Y_k, Z_{k+1}^2 \big) \in C_j \big( \Gamma_{N_{2,k+1}}^Y \big) \mid \widehat Y_k = y_i^k \Big) \\
			= & \Prob \Big( \mathcal{M}_{\widetilde b, \widetilde \sigma} \big( t_k, y_i^k, Z_{k+1}^2 \big) \in C_j \big( \Gamma_{N_{2,k+1}}^Y \big) \Big)                                  \\
			= & F_i^k \big( y_{j + 1/2}^{k+1} \big) - F_i^k \big( y_{j - 1/2}^{k+1} \big)
		\end{aligned}
	\end{equation*}
	and
	\begin{equation*}
		\begin{aligned}
			p_j^{k+1}
			= & \Prob \big( \widehat Y_{k+1} = y_j^{k+1} \big) = \sum_{i=1}^{N_{2,k}} \Prob \big( \widehat Y_{k+1} = y_j^{k+1} \mid \widehat Y_k = y_i^k \big) \Prob \big( \widehat Y_k = y_i^k \big) \\
			= & \sum_{i=1}^{N_{2,k}} p_i^{k} \, p_{ij}^k.
		\end{aligned}
	\end{equation*}
\end{proof}

As an illustration, we display in Figure \ref{RH:fig:recursive_quantization_vol_60days} the rescaled grids obtained after recursive quantization of the boosted-volatility, where $\widehat v_k = \e^{-\kappa t_k} \widehat Y_k$ and $(\widehat Y_k)_{k = 1, \dots, n}$ are the quantizers built using the fast recursive quantization approach.

\begin{figure}[!h]
	\centering
	\includegraphics[width=1.\textwidth]{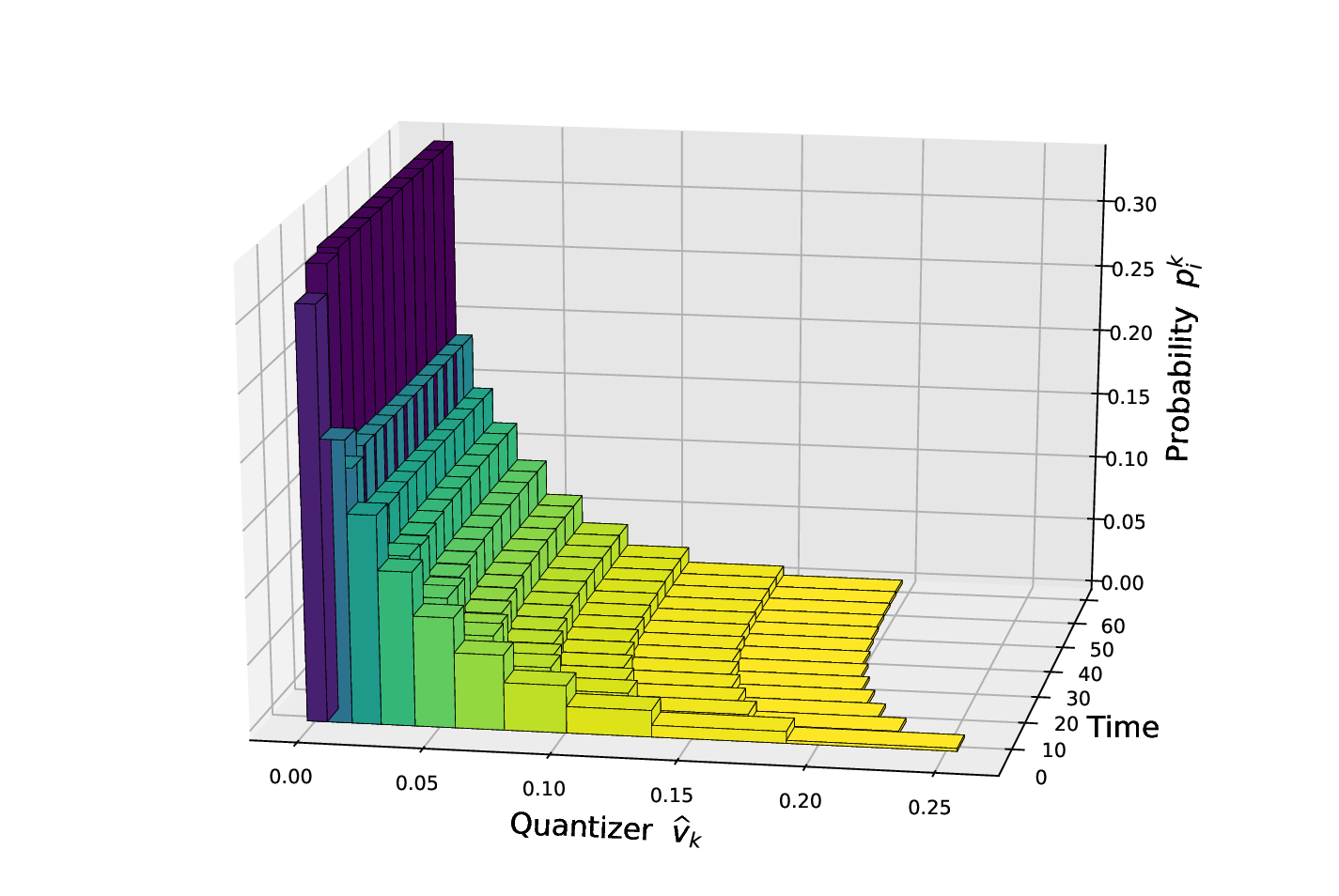}
	\caption[Rescaled Recursive quantization of the boosted-volatility process with its associated weights from $t=0$ to $t=60$ days with a time step of $5$ days with grids of size $N=10$.]{\textit{Rescaled Recursive quantization of the boosted-volatility process with its associated weights from $t=0$ to $t=60$ days with a time step of $5$ days with grids of size $N=10$. The recursive quantization methodology is applied to $\widehat Y_k$ and then we display the rescaled volatility $\widehat v_k = \e^{- \kappa t_k} \widehat Y_k$.}}
	\label{RH:fig:recursive_quantization_vol_60days}
\end{figure}

\subsubsection{Quantizing the asset (a one-dimensional case again)}

Now, using the fact that $(Y_t)_{t}$ has already been quantized and the Euler-Maruyama scheme of $(X_t)_{t}$, as defined \eqref{RH:eulerscheme}, we define the Markov quantized scheme
\begin{equation}\label{RH:recurQuantiAsset}
	\widetilde{X}_{t_{k+1}} = \mathcal{E}_{b, \sigma} \big( t_k, \widehat X_{t_k}, \widehat Y_{t_k}, Z_{k+1}^1 \big), \qquad \widehat X_{t_{k+1}} = \Proj_{\Gamma_{N_{1,k+1}}^X} \big( \widetilde{X}_{t_{k+1}} \big)
\end{equation}
where the projection operator $\Proj_{\Gamma_{N_{1,k+1}}^X}(\cdot)$ is defined in \eqref{RH:quantizOfX}, $\Gamma_{N_{1,k+1}}^X$ is the optimal $N_{1,k+1}$-quantizer of $\widetilde{X}_{t_{k+1}}$ and $Z_{k+1}^1 \sim \N(0,1)$. Again, in order to simplify the notations, $\widetilde X_{t_{k}}$ and $\widehat X_{t_k}$ are denoted in what follows by $\widetilde X_k$ and $\widehat X_k$.

Note that we are still in an one-dimensional case, hence we can apply the same methodology as developed in Appendix \ref{RH:appendix:optquant} and build recursively the quantization $\big( \widehat X_k \big)_{k = 0, \dots, n}$ as detailed above, where the $N_{1,k}$-tuple are defined by $x_{_{1:N_{1,k}}}^k = \big( x_1^k, \dots, x_{N_{1,k}}^k \big)$. Replacing $X$ by $\widetilde X_k$ in \eqref{RH:distortozero} yield
\begin{equation}
	\begin{aligned}
		x_{j_1}^{k+1}
		 & = \frac{ \E \Big[ \mathcal{E}_{b, \sigma} \big( t_k, \widehat X_{t_k}, \widehat Y_{t_k}, Z_{k+1}^1 \big) \1_{ \mathcal{E}_{b, \sigma} \big( t_k, \widehat X_{t_k}, \widehat Y_{t_k}, Z_{k+1}^1 \big) \, \in \, C_{j_1} \big( \Gamma_{N_{1,k+1}}^X \big) } \Big] }{\Prob \Big( \mathcal{E}_{b, \sigma} \big( t_k, \widehat X_{t_k}, \widehat Y_{t_k}, Z_{k+1}^1 \big) \in C_{j_1} \big( \Gamma_{N_{1,k+1}}^X \big) \Big) }                                                                                                                    \\
		 & = \frac{ \displaystyle \sum_{i_1=1}^{N_{1,k}} \sum_{i_2=1}^{N_{2,k}} \E \Big[ \mathcal{E}_{b, \sigma} \big( t_k, x_{i_1}^k, y_{i_2}^k, Z_{k+1}^1 \big) \1_{ \mathcal{E}_{b, \sigma} \big( t_k, x_{i_1}^k, y_{i_2}^k, Z_{k+1}^1 \big) \, \in \, C_{j_1} \big( \Gamma_{N_{1,k+1}}^X \big) } \Big] \, p_{(i_1,i_2)}^k}{ \displaystyle \sum_{i_1=1}^{N_{1,k}} \sum_{i_2=1}^{N_{2,k}} \Prob \Big( \mathcal{E}_{b, \sigma} \big( t_k, x_{i_1}^k, y_{i_2}^k, Z_{k+1}^1 \big) \in C_{j_1} \big( \Gamma_{N_{1,k+1}}^X \big) \Big) \, p_{(i_1,i_2)}^k} \\
		 & = \frac{ \displaystyle \sum_{i_1=1}^{N_{1,k}} \sum_{i_2=1}^{N_{2,k}} \Big( K_{(i_1,i_2)}^k \big( x_{j_1 + 1/2}^{k+1} \big) - K_{(i_1,i_2)}^k \big( x_{j_1 - 1/2}^{k+1} \big) \Big) \, p_{(i_1,i_2)}^k}{ \displaystyle \sum_{i_1=1}^{N_{1,k}} \sum_{i_2=1}^{N_{2,k}} \Big( F_{(i_1,i_2)}^k \big( x_{j_1 + 1/2}^{k+1} \big) - F_{(i_1,i_2)}^k \big( x_{j_1 - 1/2}^{k+1} \big) \Big) \, p_{(i_1,i_2)}^k}
	\end{aligned}
\end{equation}
where $p_{(i_1,i_2)}^k = \Prob \big( \widehat X_k = x_{i_1}^k, \widehat Y_k = y_{i_2}^k \big)$ and $F_{(i_1,i_2)}^k$ and $K_{(i_1,i_2)}^k$ are the cumulative distribution function and the first partial moment function of the normal distribution $\mu_{(i_1,i_2)}^k + Z_{k+1}^1 \sigma_{(i_1,i_2)}^k $ and they are defined by
\begin{equation}
	\begin{aligned}
		F_{(i_1,i_2)}^k ( x ) & = F_{_Z} \bigg( \frac{x - \mu_{(i_1,i_2)}^k }{\sigma_{(i_1,i_2)}^k} \bigg)                                                                                                                   \\
		K_{(i_1,i_2)}^k ( x ) & = \mu_{(i_1,i_2)}^k F_{_Z} \bigg( \frac{x - \mu_{(i_1,i_2)}^k }{\sigma_{(i_1,i_2)}^k} \bigg) + \sigma_{(i_1,i_2)}^k K_{_Z} \bigg( \frac{x - \mu_{(i_1,i_2)}^k }{\sigma_{(i_1,i_2)}^k} \bigg) \\
	\end{aligned}
\end{equation}
with
\begin{equation} \label{RH:mean_std_euler}
	\mu_{(i_1,i_2)}^k = x_{i_1}^k + b(t_k, x_{i_1}^k, y_{i_2}^k) h \qquad \textrm{ and } \qquad \sigma_{(i_1,i_2)}^k = \sigma(t_k, x_{i_1}^k, y_{i_2}^k) \sqrt{h}
\end{equation}
and $F_{_Z}$ and $K_{_Z}$ are the cumulative distribution function and the first partial moment of the standard normal distribution.

Finally, we apply the Lloyd method defined in Appendix \eqref{RH:LloydAlgo} with $F_{_X}$ and $K_{_X}$ defined by
\begin{equation}
	\begin{aligned}
		F_{_X} ( x ) = \sum_{i_1=1}^{N_{1,k}} \sum_{i_2=1}^{N_{2,k}} p_{(i_1,i_2)}^k \, F_{(i_1,i_2)}^k ( x ) \qquad \textrm{ and } \qquad K_{_X} ( x ) = \sum_{i_1=1}^{N_{1,k}} \sum_{i_2=1}^{N_{2,k}} p_{(i_1,i_2)}^k \, K_{(i_1,i_2)}^k ( x ).
	\end{aligned}
\end{equation}

The sensitive part concerns the computation of the joint probabilities $p_{(i_1,i_2)}^k$. Indeed, they are needed at each step in order to be able to design recursively the quantization tree.
\begin{lemme}
	The joint probabilities $p_{(i_1,i_2)}^k$ are given by the following forward induction
	\begin{equation}\label{RH:joint_proba_couple}
		p_{(j_1,j_2)}^{k+1}
		= \sum_{i=1}^{N_{1,k}} \sum_{j=1}^{N_{2,k}} p_{(i_1,i_2)}^k \Prob \big( \widehat X_{k+1} = x_{j_1}^{k+1}, \widehat Y_{k+1} = y_{j_2}^{k+1} \mid \widehat X_k = x_{i_1}^k, \widehat Y_k = y_{i_2}^k \big)
	\end{equation}
	where the joint conditional probabilities $\Prob \big( \widehat X_{k+1} = x_{j_1}^{k+1}, \widehat Y_{k+1} = y_{j_2}^{k+1} \mid \widehat X_k = x_{i_1}^k, \widehat Y_k = y_{i_2}^k \big)$ are given by the formulas below, depending on the correlation
	\begin{itemize}[wide=0pt]
		\item if $\Corr(Z_{k+1}^1, Z_{k+1}^2) = \rho = 0$
		      \begin{equation}
			      \Prob \big( \widehat X_{k+1} = x_{j_1}^{k+1}, \widehat Y_{k+1} = y_{j_2}^{k+1} \mid \widehat X_k = x_{i_1}^k, \widehat Y_k = y_{i_2}^k \big) = p_{i_2 j_2}^k \, \Big[ \N \big( x_{i_1,i_2,j_1,+}^{k} \big) - \N \big( x_{i_1,i_2,j_1,-}^{k} \big) \Big],
		      \end{equation}
		      where $p_{i_2 j_2}^k$ is defined in \eqref{RH:condiProbavol} and
		      \begin{equation}
			      x_{i_1,i_2,j_1,-}^{k} = \frac{x_{j_1-1/2}^{k+1} - \mu_{(i_1,i_2)}^k }{\sigma_{(i_1,i_2)}^k},\qquad x_{i_1,i_2,j_1,+}^{k} = \frac{x_{j_1+1/2}^{k+1} - \mu_{(i_1,i_2)}^k }{\sigma_{(i_1,i_2)}^k},
		      \end{equation}
		      with $\mu_{(i_1,i_2)}^k$ and $\sigma_{(i_1,i_2)}^k$ defined in \eqref{RH:mean_std_euler}.

		\item if $\Corr(Z_{k+1}^1, Z_{k+1}^2) = \rho \neq 0$
		      \begin{equation}\label{RH:conditionalProba}
			      \begin{aligned}
				      \Prob & \big( \widehat X_{k+1} = x_{j_1}^{k+1}, \widehat Y_{k+1} = y_{j_2}^{k+1} \mid \widehat X_k = x_{i_1}^k, \widehat Y_k = y_{i_2}^k \big)                                                                                  \\
				            & = \Prob \Big( Z_{k+1}^1 \in \big( x_{i_1,i_2,j_1,-}^{k}, x_{i_1,i_2,j_1,+}^{k} \big], Z_{k+1}^2 \in \Big( \sqrt{y_{i_2,j_2,-}^{k}} - \lambda_{i_2}^k, \sqrt{y_{i_2,j_2,+}^{k}} - \lambda_{i_2}^k \Big] \Big)            \\
				            & \qquad + \Prob \Big( Z_{k+1}^1 \in \big( x_{i_1,i_2,j_1,-}^{k}, x_{i_1,i_2,j_1,+}^{k} \big], Z_{k+1}^2 \in \Big[ - \sqrt{y_{i_2,j_2,+}^{k}} - \lambda_{i_2}^k, - \sqrt{y_{i_2,j_2,-}^{k}} - \lambda_{i_2}^k \Big) \Big) \\
			      \end{aligned}
		      \end{equation}
		      where
		      \begin{equation}
			      y_{i_2,j_2,-}^{k} = 0 \vee \frac{y_{j_2-1/2}^{k+1} - \mu_{i_2}^k}{\kappa_{i_2}^k},\qquad y_{i_2,j_2,+}^{k} = 0 \vee \frac{y_{j_2+1/2}^{k+1} - \mu_{i_2}^k}{\kappa_{i_2}^k},
		      \end{equation}
		      with $\mu_{i_2}^k$, $\kappa_{i_2}^k$ and $\lambda_{i_2}^k$ defined in \eqref{RH:mean_std_milstein}.

	\end{itemize}
\end{lemme}

\begin{remark}
	The probability in the right hand side of \eqref{RH:conditionalProba} can be computed using the cumulative distribution function of a correlated bivariate normal distribution\footnote{C++ implementation of the upper right tail of a bivariate normal distribution can be found in John Burkardt's website \url{https://people.sc.fsu.edu/~jburkardt/cpp_src/toms462/toms462.html}.}. Indeed, let
	$$F_{\rho}(x_1, x_2) = \Prob (X_1 \leq x_1, X_2 \leq x_2)$$
	the cumulative distribution function of the correlated centered Gaussian vector $(X_1, X_2)$ with unit variance and correlation $\rho$, we have
	\begin{equation}
		\Prob \big( X_1 \in [a,b], X_2 \in [c,d] \big) = F_{\rho}(b, d) - F_{\rho}(b, c) - F_{\rho}(a, d) + F_{\rho}(a, c)
	\end{equation}
	with $a,c \geq -\infty$ and $b,d \leq + \infty$.
\end{remark}

\begin{proof}
	\begin{equation*}
		\begin{aligned}
			p_{(j_1,j_2)}^{k+1}
			= & \Prob \big( \widehat X_{k+1} = x_{j_1}^{k+1}, \widehat Y_{k+1} = y_{j_2}^{k+1} \big)                                                                                                                                                                        \\
			= & \sum_{i=1}^{N_{1,k}} \sum_{j=1}^{N_{2,k}} \Prob \big( \widehat X_{k+1} = x_{j_1}^{k+1}, \widehat Y_{k+1} = y_{j_2}^{k+1} \mid \widehat X_k = x_{i_1}^k, \widehat Y_k = y_{i_2}^k \big) \Prob \big( \widehat X_k = x_{i_1}^k, \widehat Y_k = y_{i_2}^k \big) \\
			= & \sum_{i=1}^{N_{1,k}} \sum_{j=1}^{N_{2,k}} p_{(i_1,i_2)}^k \Prob \big( \widehat X_{k+1} = x_{j_1}^{k+1}, \widehat Y_{k+1} = y_{j_2}^{k+1} \mid \widehat X_k = x_{i_1}^k, \widehat Y_k = y_{i_2}^k \big).                                                     \\
		\end{aligned}
	\end{equation*}

	\begin{itemize}[wide=0pt]
		\item if $\Corr(Z_{k+1}^1, Z_{k+1}^2) = \rho = 0$
		      \begin{equation*}
			      \begin{aligned}
				       & \Prob \big( \widehat X_{k+1} = x_{j_1}^{k+1}, \widehat Y_{k+1} = y_{j_2}^{k+1} \mid \widehat X_k = x_{i_1}^k, \widehat Y_k = y_{i_2}^k \big)                                                             \\
				       & \qquad \qquad \qquad \qquad \qquad \qquad= p_{i_2 j_2}^k \Prob \big( \widehat X_{k+1} = x_{j_1}^{k+1} \mid \widehat X_k = x_{i_1}^k, \widehat Y_k = y_{i_2}^k \big)                                      \\
				       & \qquad \qquad \qquad \qquad \qquad \qquad= p_{i_2 j_2}^k \Prob \Big( \widebar X_{k+1} \in \big( x_{j_1-1/2}^{k+1}, x_{j_1+1/2}^{k+1} \big] \mid \widehat X_k = x_{i_1}^k, \widehat Y_k = y_{i_2}^k \Big) \\
				       & \qquad \qquad \qquad \qquad \qquad \qquad= p_{i_2 j_2}^k \Prob \Big( \mathcal{E}_{b, \sigma} \big( t_k, x_{i_1}^k, y_{i_2}^k, Z_{k+1}^1 \big) \in \big( x_{j_1-1/2}^{k+1}, x_{j_1+1/2}^{k+1} \big] \Big) \\
				       & \qquad \qquad \qquad \qquad \qquad \qquad= p_{i_2 j_2}^k \, \Big[ \N \big( x_{i_1,i_2,j_1,+}^{k} \big) - \N \big( x_{i_1,i_2,j_1,-}^{k} \big) \Big],
			      \end{aligned}
		      \end{equation*}

		\item if $\Corr(Z_{k+1}^1, Z_{k+1}^2) = \rho \neq 0$
		      \begin{equation*}
			      \begin{aligned}
				      \Prob & \big( \widehat X_{k+1} = x_{j_1}^{k+1}, \widehat Y_{k+1} = y_{j_2}^{k+1} \mid \widehat X_k = x_{i_1}^k, \widehat Y_k = y_{i_2}^k \big)                                                                                                                                                      \\
				            & = \Prob \Big( \mathcal{E}_{b, \sigma} \big( t_k, x_{i_1}^k, y_{i_2}^k, Z_{k+1}^1 \big) \in \big( x_{j_1-1/2}^{k+1}, x_{j_1+1/2}^{k+1} \big], \mathcal{M}_{\widetilde b, \widetilde \sigma} \big( t_k, y_{i_2}^k, Z_{k+1}^2 \big) \in \big( y_{j_2-1/2}^{k+1}, y_{j_2+1/2}^{k+1} \big] \Big) \\
				            & = \Prob \Big( \mu_{(i_1,i_2)}^k + \sigma_{(i_1,i_2)}^k Z_{k+1}^1 \in \big( x_{j_1-1/2}^{k+1}, x_{j_1+1/2}^{k+1} \big], \mu_{i_2}^k + \kappa_{i_2}^k ( Z_{k+1}^2 + \lambda_{i_2}^k )^2 \in \big( y_{j_2-1/2}^{k+1}, y_{j_2+1/2}^{k+1} \big] \Big)                                            \\
				            & = \Prob \Big( Z_{k+1}^1 \in \big( x_{i_1,i_2,j_1,-}^{k}, x_{i_1,i_2,j_1,+}^{k} \big], ( Z_{k+1}^2 + \lambda_{i_2}^k )^2 \in \big( y_{i_2,j_2,-}^{k}, y_{i_2,j_2,+}^{k} \big] \Big)                                                                                                          \\
				            & = \Prob \Big( Z_{k+1}^1 \in \big( x_{i_1,i_2,j_1,-}^{k}, x_{i_1,i_2,j_1,+}^{k} \big], Z_{k+1}^2 \in \Big( \sqrt{y_{i_2,j_2,-}^{k}} - \lambda_{i_2}^k, \sqrt{y_{i_2,j_2,+}^{k}} - \lambda_{i_2}^k \Big] \Big)                                                                                \\
				            & \qquad + \Prob \Big( Z_{k+1}^1 \in \big( x_{i_1,i_2,j_1,-}^{k}, x_{i_1,i_2,j_1,+}^{k} \big], Z_{k+1}^2 \in \Big[ - \sqrt{y_{i_2,j_2,+}^{k}} - \lambda_{i_2}^k, - \sqrt{y_{i_2,j_2,-}^{k}} - \lambda_{i_2}^k \Big) \Big).
			      \end{aligned}
		      \end{equation*}
	\end{itemize}
\end{proof}

\begin{remark}
	Another possibility for the quantization of the Stationary Heston model could be to use optimal quantizers for the volatility at each date $t_k$ in place of using recursive quantization. Indeed, the volatility $(v_t)_t$ being stationary and the fact that we required the volatility to start at time $0$ from the invariant measure, we could use the grid of the optimal quantization $\widehat v_0$ of size $N$ of the stationary measure with its associated weights for every dates, hence setting $\widehat v_k = \widehat v_0$. We need as well the transitions from time $t_k$ to $t_{k+1}$ defined by
	\begin{equation}
		\Prob \big( \widehat v_{k+1} = v_{j_2}^{k+1} \mid \widehat v_k = v_{i_2}^k \big).
	\end{equation}
	These probabilities can be computed using the conditional law of the CIR process described in \cite{cox2005theory, andersen2007efficient}, which is a non-central chi-square distribution.
	Then, we would build the recursive quantizer of the $\log$-asset at date $\widehat X_{k+1}$ with the standard methodology of recursive quantization using the already built quantizers of the volatility $\widehat v_k$ and the $\log$-asset $\widehat X_k$ at time $t_k$, i.e.
	\begin{equation}
		\widetilde X_{k+1} = \mathcal{E}_{b, \sigma} \big( t_k, \widehat X_k, \widehat v_k, Z_{k+1}^1 \big) \quad \mbox{and} \quad \widehat X_{k+1} = \Proj_{\Gamma_{N_{1,k+1}}^X} \big( \widetilde{X}_{k+1} \big)
	\end{equation}
	where, this time, the Euler scheme is not defined in function of the boosted-volatility but directly in function of the volatility and is given by
	\begin{equation}
		\mathcal{E}_{b, \sigma} \big( t, x, v, z \big) = x + h \Big( r - q - \frac{v}{2} \Big) + \sqrt{v} \sqrt{h} z.
	\end{equation}

	\medskip
	However, the difficulties with this approach come from the computation of the couple transitions
	\begin{equation}
		\Prob \big( \widehat X_{k+1} = x_{j_1}^{k+1}, \widehat v_{k+1} = v_{j_2}^{k+1} \mid \widehat X_k = x_{i_1}^k, \widehat v_k = v_{i_2}^k \big).
	\end{equation}
	Indeed, these probability weights would not be as straightforward to compute as the methodology we adopt in this paper, namely using time-discretization schemes for both components. Our approach allows us to express the conditional probability of the couple as the probability that a correlated bi-variate Gaussian vector lies in a rectangle domain and this can be easily be computed numerically.

\end{remark}

\subsubsection{About the $L^2$-error}

In this part, we study the $L^2$-error induced by the product recursive quantization approximation $\widehat U_k = (\widehat X_k, \widehat Y_k )$ of $\widebar U_k = (\widebar X_k, \widebar Y_k)$, the time-discretized processes defined in \eqref{RH:discretized_boostedvol} and \eqref{RH:discretized_couple} by
\begin{equation}\label{RH:timescheme_U}
	\widebar U_k = F_{k-1} ( \widebar U_{k-1}, Z_k)
\end{equation}
where $Z_k = (Z_k^1, Z_k^2)$ is a standardized correlated Gaussian vector and the hybrid discretization scheme $F_k(u, Z)$ is given by
\begin{equation}
	F_k(u, Z) = \Bigg(
	\begin{aligned}
		 & \mathcal{E}_{b, \sigma} \big( t_k, x, y, Z_{k+1}^1 \big)                    \\
		 & \mathcal{M}_{\widetilde b, \widetilde \sigma} \big( t_k, y, Z_{k+1}^2 \big)
	\end{aligned} \Bigg).
\end{equation}

We recall the definition of the product recursive quantizer $\widehat U_k = (\widehat X_k, \widehat Y_k)$. Its first component $\widehat X_k$ is the projection of $\widetilde X_k$ onto $\Gamma_{N_{1,k}}^X$ and the second component $\widehat Y_k$ is the projection of $\widetilde Y_k$ onto $\Gamma_{N_{2,k}}^Y$, i.e.,
\begin{equation}\label{RH:recur_quantization_U}
	\widehat X_{k+1} = \Proj_{\Gamma_{N_{1,k+1}}^X} \big( \widetilde{X}_{k+1} \big) \quad \mbox{and} \quad \widehat Y_{k+1} = \Proj_{\Gamma_{N_{2,k+1}}^Y} \big( \widetilde{Y}_{k+1} \big)
\end{equation}
where $\widetilde X_k$ and $\widetilde Y_k$ are defined in \eqref{RH:recurQuantiVol} and \eqref{RH:recurQuantiAsset}, respectively. Moreover, if we consider the couple $\widetilde U_k = (\widetilde X_k, \widetilde Y_k)$, using the above notations we have
\begin{equation}\label{RH:diffusion_quantized_schem_U}
	\widetilde U_k = F_{k-1} ( \widehat U_{k-1}, Z_k ).
\end{equation}


It has been shown in \cite{abbas2018product,sagna2018general} that if, for all $k=0, \dots, n-1$, the schemes $F_{k}(u,z)$ are Lipschitz in $u$, then there exists constants $j=1, \dots, n, \, C_j < + \infty$ such that
\begin{equation}
	\Vert \widehat U_k - \widebar U_k \Vert_{_2} \leq \sum_{j=1}^k C_j \big( N_{1,j} \times N_{2,j} \big)^{-1/2}
\end{equation}
where $\widehat U_k$ and $\widebar U_k$ are the processes defined in \eqref{RH:recur_quantization_U} and \eqref{RH:diffusion_quantized_schem_U}. The proof of this result is based on the extension of Pierce's lemma to the case of product quantization (see Lemma 2.3 in \cite{sagna2018general}).

\medskip
In our case, the diffusion of the boosted volatility in the CIR model does not have Lipschitz drift and volatility components, hence the above result from \cite{abbas2018product,sagna2018general} does not apply in our context. Even if we can hope to obtain similar results by applying the same kind of arguments, the results we obtain have to considered carefully. Indeed, when we take the limit in $n \rightarrow + \infty$, the number of time-step, the error upper-bound term goes to infinity. However, in practice, we consider $h = kT/n$ fixed and then study the behavior of $\widehat U_k$ in function of $N_{1,j}$ and $N_{2,j}$ for $j \geq k$. The proof of the following proposition is given in Appendix \ref{RH:appendix:proofl2error}.

\begin{proposition} \label{RH:prop:l2-error}
	Let $b$, $\sigma$, $\widetilde b$ and $\widetilde \sigma$, defined by \eqref{RH:drift_and_vol_boostedvol} and \eqref{RH:drift_and_vol_logasset}, the coefficients of the $\log$-asset and the boosted-volatility of the Heston model. Let, for every $k=0, \dots, n$, $\widehat U_k$ the hybrid recursive product quantizer at level $N_{1,k} \times N_{2,k}$ of $\widebar U_k$. Then, for every $k=0, \dots, n$
	\begin{equation}
		\begin{aligned}
			\Vert \widehat U_k - \widebar U_k \Vert_{_2}
			 & \leq \sum_{j=0}^k \widetilde A_{j,k} \big( N_{1,j} \times N_{2,j} \big)^{-1/2} + B_{k} \sqrt{h} \\
		\end{aligned}
	\end{equation}
	where
	\begin{equation}
		\widetilde A_{j,k} = 2^{\frac{p-2}{2p}} C_p^2 A_{j,k} \bigg( 2^{(\frac{p}{2}-1 )j} \beta_{p}^j \Vert \widehat U_0 \Vert_{_2}^p + \alpha_{p} \frac{1 - 2^{(\frac{p}{2}-1 )j} \beta_{p}^j }{1 - 2^{\frac{p}{2}-1} \beta_{p} } \bigg)^{1/p}
	\end{equation}
	with
	\begin{equation}
		A_{j,k} = 2^{\frac{k-j}{2}} \e^{\frac{\sqrt{h}}{2}(k-j)} \quad \mbox{and} \quad B_k = C_{T} (h) \sum_{j=0}^{k-1} 2^{\frac{k-1-j}{2}} \e^{\frac{\sqrt{h}}{2}(k-1-j)}
	\end{equation}
	where $\sum_{\emptyset} = 0$ by convention and $C_T(h) = O(1)$.
\end{proposition}

\subsection{Backward algorithm for Bermudan and Barrier options}

\paragraph{Bermudan Options}

A Bermudan option is a financial derivative product that gives the right to its owner to buy or sell (or to enter to, in the case of a swap) an underlying product with a given payoff $\psi_t(\cdot, \cdot)$ at predefined exercise dates $\{t_0, \cdots, t_n\}$. Its price, at time $t_0=0$, is given by
\begin{equation*}
	\sup_{\tau \in \{ t_0, \cdots, t_n \}} \E \Big[ \e^{ - r \tau } \psi_{\tau} ( X_{\tau}, Y_{\tau} ) \mid \F_{t_0} \Big]
\end{equation*}
where $X_t$ and $Y_t$ are solutions to the system defined in \eqref{RH:systemvolsto}.

In this part, we follow the numerical solution first introduced by \cite{printems2005quantization,bally2003quantization}. They proposed to solve discrete-time optimal stopping problems using a quantization tree of the risk factors $X_t$ and $Y_t$.

Let $\F^{X,Y} = ( \F )_{0 \leq k \leq n}$ the natural filtration of $X$ and $Y$. Hence, we can define recursively the sequence of random variable $L^p$-integrable $(V_k)_{0 \leq k \leq n}$
\begin{equation}\label{RH:BDPP}
	\left\{
	\begin{aligned}
		 & V_n = \e^{ - r t_n } \psi_n( X_n, Y_n ),                                                                            \\
		 & V_k = \max \big( \e^{ - r t_k } \psi_k(X_k, Y_k), \E [ V_{k+1} \mid \F_k ] \big) \mathrm{,\qquad} 0 \leq k \leq n-1
	\end{aligned} \right.
\end{equation}
called \textit{Backward Dynamic Programming Principle}. Then
\begin{equation*}
	V_0 = \sup \big\{ \E [ \e^{ - r \tau } \psi_\tau( X_\tau, Y_\tau ) \mid \F_0 ], \tau \in \Theta_{0,n} \big\}
\end{equation*}
with $\Theta_{0,n}$ the set of all stopping times taking values in $\{ t_0, \cdots, t_n \}$.
The sequence $(V_k)_{0 \leq k \leq n}$ is also known as the Snell envelope of the obstacle process $\big( \e^{ - r t_k } \psi_k(X_k, Y_k) \big)_{0 \leq k \leq n}$. In the end, $\E [ V_0 ]$ is the quantity we are interested in. Indeed, $\E [ V_0 ]$ is the price of the Bermudan option whose payoff is $\psi_k$ and is exercisable at dates $\{ t_1, \cdots, t_n \}$.

Following what was defined in \eqref{RH:BDPP}, in order to compute $\E [ V_0 ]$, we will need to use the previously defined quantizer of $X_k$ and $Y_k$: $\widehat X_k$ and $\widehat Y_k$. Hence, for a given global budget $N = N_{1,0} N_{2,0} + \cdots + N_{1,n} N_{2,n}$, the total number of nodes of the tree by the couple $(\widehat X_k, \widehat Y_k)_{0 \leq k \leq n}$, we can approximate the \textit{Backward Dynamic Programming Principle} \eqref{RH:BDPP} by the following sequence involving the couple $(\widehat{X}_k, \widehat Y_k )_{0 \leq k \leq n}$
\begin{equation} \label{RH:BDPP_Quantif}
	\left\{ \begin{aligned}
		 & \widehat V_n = \e^{-r t_n} \psi_n ( \widehat X_n, \widehat Y_n ),                                                                                                     \\
		 & \widehat V_k = \max \big( \e^{-r t_k} \psi_k(\widehat X_k, \widehat Y_k), \E [ \widehat V_{k+1} \mid ( \widehat X_k, \widehat Y_k ) ] \big),\qquad k = 0, \dots, n-1.
	\end{aligned} \right.
\end{equation}
\begin{remark}
	A direct consequence of choosing recursive Markovian Quantization to spatially discretize the problem is that the sequence $(\widehat{X}_k, \widehat Y_k )_{0 \leq k \leq n}$ is Markovian. Hence $(\widehat V_k)_{0 \leq k \leq n}$ defined in \eqref{RH:BDPP_Quantif} obeying a \textit{Backward Dynamic Programming Principle} is the Snell envelope of $\big( \e^{ - r t_k } \psi_k ( \widehat X_k, \widehat Y_k ) \big)_{0 \leq k \leq n}$. This is the main difference with the first approach of \cite{printems2005quantization,bally2003quantization}, where in there case they only had a pseudo-Snell envelope of $\big( \e^{ - r t_k } \psi_k ( \widehat X_k, \widehat Y_k ) \big)_{0 \leq k \leq n}$.
\end{remark}

\medskip
Using the discrete feature of the quantizers, \eqref{RH:BDPP_Quantif} can be rewritten
\begin{equation} \label{RH:BDPP_Quantif_num} \left\{
	\begin{aligned}
		 & \widehat v_n ( x_{i_1}^n, y_{i_2}^n ) = \e^{ - r t_n }  \psi_n ( x_{i_1}^n, y_{i_2}^n )
		,\quad
		\begin{matrix}
			i_1 = 1, \dots, N_{1,n} \\
			i_2 = 1, \dots, N_{2,n}
		\end{matrix}
		\\
		 & \widehat v_k ( x_{i_1}^k, y_{i_2}^k ) = \max \Big( \e^{ - r t_k } \psi_k ( x_{i_1}^k, y_{i_2}^k ), \sum_{j_1 = 1}^{N_{1,k+1}} \sum_{j_2 = 1}^{N_{2,k+1}} \pi_{(i_1,i_2),(j_1,j_2)}^k \widehat{v}_{k+1} ( x_{j_1}^{k+1}, y_{j_2}^{k+1} ) \Big)
		, \quad
		\begin{matrix}
			k = 0, \dots, n-1       \\
			i_1 = 1, \dots, N_{1,k} \\
			i_2 = 1, \dots, N_{2,k}
		\end{matrix}
	\end{aligned} \right.
\end{equation}
where $ \pi_{(i_1,i_2),(j_1,j_2)}^k = \Prob \big( \widehat X_{k+1} = x_{j_1}^{k+1}, \widehat Y_{k+1} = y_{j_2}^{k+1} \mid \widehat X_k = x_{i_1}^k, \widehat Y_k = y_{i_2}^k \big)$ is the conditional probability weight given in \eqref{RH:conditionalProba}. Finally, the approximation of the price of the Bermudan option is given by
\begin{equation}
	\E \big[ \widehat v_0 ( x_0, \widehat Y_0 ) \big] = \sum_{i=1}^{N_{2,0}} p_i \, \widehat v_0 ( x_0, y_i^0 )
\end{equation}
with $p_i = \Prob \big( \widehat Y_0 = y_i^0 \big)$ given by \eqref{RH:probainitialmeasure}.

\paragraph{Barrier Options}

A Barrier option is a path-dependent financial product whose payoff at maturity date $T$ depends on the value of the process $X_T$ at time $T$ and its maximum or minimum over the period $[0,T]$. More precisely, we are interested by options with the following types of payoff $h$
\begin{equation}
	h = f(X_T) \1_{ \{ \sup_{t \in [0,T]} X_t \in I \} }\qquad or \qquad h = f(X_T) \1_{ \{ \inf_{t \in [0,T]} X_t \in I \} }
\end{equation}
where $I$ is an unbounded interval of $\R$, usually of the forme $(- \infty, L]$ or $[L, + \infty)$ ($L$ is the barrier) and $f$ can be any vanilla payoff function (Call, Put, Spread, Butterfly, ...).

\medskip
In this part, we follow the methodology initiated in \cite{sagna2010pricing} in the case of functional quantization. This work is based on the Brownian bridge method applied to the Euler-Maruyama scheme as described e.g. in \cite{pages2018numerical}. We generalize it to stochastic volatility models and product Markovian recursive quantization. $X_t$ being discretized by an Euler-Maruyama scheme, yielding $\widebar X_k$ with $k=0, \dots, n$, we can determine the law of $\max_{t \in [0,T]} \widebar X_t$ and $\min_{t \in [0,T]} \widebar X_t$ given the values $\widebar X_k = x_k, \widebar Y_k = y_k, k = 0, \dots, n$
\begin{equation}
	\Law \Big( \max_{t \in [0,T]} \widebar X_t \mid \widebar X_k = x_k, \widebar Y_k = y_k, k = 0, \dots, n \Big) = \Law \Big( \max_{k = 0, \dots, n-1} \big( G_{(x_k,y_k), x_{k+1}}^k \big)^{-1} (U_k) \Big)
\end{equation}
and
\begin{equation}
	\Law \Big( \min_{t \in [0,T]} \widebar X_t \mid \widebar X_k = x_k, \widebar Y_k = y_k, k = 0, \dots, n \Big) = \Law \Big( \max_{k = 0, \dots, n-1} \big( F_{(x_k,y_k), x_{k+1}}^k \big)^{-1} (U_k) \Big)
\end{equation}
where $(U_k)_{k=0,\dots, n-1}$ are i.i.d uniformly distributed random variables over the unit interval and $(G_{(x,y),z}^k)^{-1}$ and $(F_{(x,y),z}^k)^{-1}$ are the inverse of the conditional distribution functions $G_{(x,y),z}^k$ and $F_{(x,y),z}^k$ defined by
\begin{equation}\label{RH:upoutfunction}
	G_{(x,y),z}^k(u) = \Big( 1 - \e^{-2n \frac{(x-u)(z-u)}{T \sigma^2(t_k,x,y)}} \Big) \1_{\{ u \geq \max(x,z) \}}
\end{equation}
and
\begin{equation}\label{RH:downoutfunction}
	F_{(x,y),z}^k(u) = 1 - \Big( 1 - \e^{-2n \frac{(x-u)(z-u)}{T \sigma^2(t_k,x,y)}} \Big) \1_{\{ u \leq \min(x,z) \}}.
\end{equation}

Now, using the resulting representation formula for $\E f (\widebar X_T, \max_{t \in [0,T]} \widebar X_t)$ (see e.g. \cite{sagna2010pricing, pages2018numerical}), we have a new representation formula for the price of up-and-out options $\widebar P_{UO}$ and down-and-out options $\widebar P_{DO}$
\begin{equation}
	\widebar P_{UO} = \e^{-rT} \E \big[ f(\widebar X_T) \1_{\sup_{t \in [0,T]} \widebar X_t \leq L} \big] = \e^{-rT} \E \bigg[ f(\widebar X_T) \prod_{k=0}^{n-1} G_{ (\overline{X}_k, \overline{Y}_k), \widebar X_{k+1}}^k (L) \bigg]
\end{equation}
and
\begin{equation}
	\widebar P_{DO} = \e^{-rT} \E \big[ f(\widebar X_T) \1_{\inf_{t \in [0,T]} \widebar X_t \geq L} \big] = \e^{-rT} \E \bigg[ f(\widebar X_T) \prod_{k=0}^{n-1} \Big( 1 - F_{(\widebar X_k, \widebar Y_k), \widebar X_{k+1}}^k (L) \Big) \bigg]
\end{equation}
where $L$ is the barrier.

\medskip
Finally, replace $\widebar X_k$ and $\widebar Y_k$ by $\widehat X_k$ and $\widehat Y_k$ and apply the recursive algorithm in order to approximate $\widebar P_{UO}$ or $\widebar P_{DO}$ by $\E [ \widehat V_0 ]$ or equivalently $\E [ \widehat v_0 (x_0, \widehat Y_0) ]$
\begin{equation} \label{RH:barrier_Quantif}
	\left\{ \begin{aligned}
		 & \widehat V_n = \e^{-rT} f ( \widehat X_n ),                                                                                                                          \\
		 & \widehat V_k = \E \big[ g_k(\widehat X_k,\widehat Y_k, \widehat X_{k+1}) \widehat V_{k+1} \mid (\widehat X_k, \widehat Y_k) \big] \mathrm{,\qquad} 0 \leq k \leq n-1
	\end{aligned} \right.
\end{equation}
that can be rewritten
\begin{equation} \label{RH:barrier_Quantif_num} \left\{
	\begin{aligned}
		 & \widehat v_n ( x_{i_1}^n, y_{i_2}^n ) = \e^{-rT} f ( x_i^n )
		,\quad
		\begin{matrix}
			i = 1, \dots, N_{1,n} \\
			j = 1, \dots, N_{2,n}
		\end{matrix}
		\\
		 & \widehat v_k ( x_{i_1}^k, y_{i_2}^n ) = \sum_{j_1 = 1}^{N_{1,k+1}} \sum_{j_2 = 1}^{N_{2,k+1}} \pi_{(i_1,i_2),(j_1,j_2)}^k \widehat{v}_{k+1} ( x_{j_1}^{k+1}, y_{j_2}^{k+1} ) g_k( x_{i_1}^k, y_{i_2}^k, x_{j_1}^{k+1} )
		, \quad
		\begin{matrix}
			k = 0, \dots, n-1     \\
			i = 1, \dots, N_{1,k} \\
			j = 1, \dots, N_{2,k}
		\end{matrix}
	\end{aligned} \right.
\end{equation}
with $ \pi_{(i_1,i_2),(j_1,j_2)}^k = \Prob \big( \widehat X_{k+1} = x_{j_1}^{k+1}, \widehat Y_{k+1} = y_{j_2}^{k+1} \mid \widehat X_k = x_{i_1}^k, \widehat Y_k = y_{i_2}^k \big)$ the conditional probabilities given in \eqref{RH:conditionalProba} and $g_k(x,y,z)$ is either equal to $G_{(x,y),z}^k(L)$ or $1- F_{(x,y),z}^k(L)$ depending on the option type. Finally, the approximation of the price of the barrier option is given by
\begin{equation}
	\E [ \widehat V_0 ] = \E \big[ \widehat v_0 ( x_0, \widehat Y_0 ) \big] = \sum_{i=1}^{N_{2,0}} p_i \, \widehat v_0 ( x_0, y_i^0 )
\end{equation}
with $p_i = \Prob \big( \widehat Y_0 = y_i^0 \big)$ given by \eqref{RH:probainitialmeasure}.

\subsection{Numerical illustrations}

In this part, we deal with numerical experiments in the Stationary Heston model. We will apply the methodology based on hybrid product recursive quantization to the pricing of European, Bermudan and Barrier options. For the model parameters, we consider the parameters given in Table \ref{RH:tab:params_withpena} obtained after the penalized calibration procedure and instead of considering the market value for $S_0$, we take $S_0=100$ in order to get prices of an order we are used to. For the size of the quantization grids, we consider grids of constant size for all time-steps: for all $k=0, \dots, n$, we take $N_{1,k} = N_1$ and $N_{2,k} = N_2$ where $n$ is the number of time steps. During the numerical tests, we vary the tuple values $(n,N_1,N_2)$.

All the numerical tests have been carried out in C++ on a laptop with a 2,4 GHz 8-Core Intel Core i9 CPU. The computations of the transition probabilities are parallelized on the CPU.

\paragraph{European options}
First, we compare, in Table \ref{RH:tab:price_european_n_180}, the price of European options with maturity $t_n = T=0.5$ ($6$ months) computed using the quantization tree to the benchmark price computed using the methodology based on the quadrature formula (the quadrature formula with Laguerre polynomials) explained in Section \ref{RH:section:pricing_calibrationEU}. In place of using the backward algorithm \eqref{RH:BDPP_Quantif} (without the function max) for computing the expectation at the expiry date, we use the weights $p_{(i_1,i_2)}^k$ defined in \eqref{RH:joint_proba_couple} and built by forward induction, in order to compute
\begin{equation}
	\E \big[ \e^{-r t_n} \psi_n ( \widehat X_n, \widehat Y_n ) \big] = \e^{-r t_n} \sum_{i_1=1}^{N_{1,n}} \sum_{i_2=1}^{N_{2,n}} \psi_n ( x_{i_1}^n, y_{i_2}^n ).
\end{equation}
We give, in parenthesis, the relative error induced by the quantization-based approximation. We compare the behavior of the pricers with different size of grids and numbers of discretization steps. We notice that the main part of the error is explained by the size of the time-step $n$.

\begin{table}[H]
	\centering
	\begin{tabular}{c|lcccccc}
		\toprule
		                        &       &           & \multicolumn{4}{c}{$(N_1,N_2)$}                                                                                 \\ \midrule
		                        & $K$   & Benchmark & $(20,5)$                        & $(50,10)$               & $(100,10)$              & $(150,10)$              & \\ \midrule \midrule
		\multirow{5}{2em}{Call} & $80$  & $20.17$   & $19.68 \, \, (2.46 \%)$         & $19.99 \, \, (0.92 \%)$ & $20.04 \, \, (0.64 \%)$ & $20.06 \, \, (0.57 \%)$ & \\
		                        & $85$  & $15.56$   & $14.97 \, \, (3.75 \%)$         & $15.35 \, \, (1.31 \%)$ & $15.42 \, \, (0.89 \%)$ & $15.43 \, \, (0.79 \%)$ & \\
		                        & $90$  & $11.24$   & $10.60 \, \, (5.68 \%)$         & $11.03 \, \, (1.84 \%)$ & $11.10 \, \, (1.18 \%)$ & $11.12 \, \, (1.02 \%)$ & \\
		                        & $95$  & $7.383$   & $6.781 \, \, (8.14 \%)$         & $7.202 \, \, (2.44 \%)$ & $7.286 \, \, (1.30 \%)$ & $7.306 \, \, (1.03 \%)$ & \\
		                        & $100$ & $4.196$   & $3.727 \, \, (11.1 \%)$         & $4.081 \, \, (2.73 \%)$ & $4.173 \, \, (0.54 \%)$ & $4.194 \, \, (0.04 \%)$ & \\ \midrule
		\multirow{6}{2em}{Put}  & $100$ & $4.469$   & $4.160 \, \, (6.90 \%)$         & $4.396 \, \, (1.61 \%)$ & $4.459 \, \, (0.22 \%)$ & $4.472 \, \, (0.08 \%)$ & \\
		                        & $105$ & $7.171$   & $7.034 \, \, (1.91 \%)$         & $7.178 \, \, (0.09 \%)$ & $7.244 \, \, (1.01 \%)$ & $7.257 \, \, (1.19 \%)$ & \\
		                        & $110$ & $10.86$   & $10.84 \, \, (0.18 \%)$         & $10.91 \, \, (0.46 \%)$ & $10.97 \, \, (1.02 \%)$ & $10.98 \, \, (1.11 \%)$ & \\
		                        & $115$ & $15.38$   & $15.43 \, \, (0.33 \%)$         & $15.40 \, \, (0.12 \%)$ & $15.43 \, \, (0.37 \%)$ & $15.44 \, \, (0.41 \%)$ & \\
		                        & $120$ & $20.30$   & $20.43 \, \, (0.60 \%)$         & $20.31 \, \, (0.02 \%)$ & $20.29 \, \, (0.05 \%)$ & $20.29 \, \, (0.04 \%)$ & \\ \bottomrule
		                        & Time  &           & $2.6$s                          & $39$s                   & $192$s                  & $480$s                  & \\ \bottomrule
	\end{tabular}
	\caption[Pricing of European options in a Stationary Heston model with product hybrid recursive quantization with time-step $n=180$.]{\textit{Comparison between European options prices, with maturity $T=0.5$ (6 months), given by quantization and the benchmark, in function of the strike $K$ and $(N_1,N_2)$ where we set $n=180$.}}
	\label{RH:tab:price_european_n_180}
\end{table}

\begin{table}[H]
	\centering
	\begin{tabular}{c|lcccccc}
		\toprule
		                        &       &           & \multicolumn{4}{c}{$n$}                                                                                 \\ \midrule
		                        & $K$   & Benchmark & $30$                    & $60$                    & $90$                    & $180$                   & \\ \midrule \midrule
		\multirow{5}{2em}{Call} & $80$  & $20.17$   & $20.00 \, \, (0.83\%)$  & $20.03 \, \, (0.70 \%)$ & $20.03 \, \, (0.72 \%)$ & $19.99 \, \, (0.92 \%)$ & \\
		                        & $85$  & $15.56$   & $15.33 \, \, (1.47\%)$  & $15.38 \, \, (1.11 \%)$ & $15.39 \, \, (1.07 \%)$ & $15.35 \, \, (1.31 \%)$ & \\
		                        & $90$  & $11.24$   & $10.94 \, \, (2.60\%)$  & $11.04 \, \, (1.78 \%)$ & $11.05 \, \, (1.63 \%)$ & $11.03 \, \, (1.84 \%)$ & \\
		                        & $95$  & $7.383$   & $7.045 \, \, (4.57\%)$  & $7.170 \, \, (2.87 \%)$ & $7.203 \, \, (2.43 \%)$ & $7.202 \, \, (2.44 \%)$ & \\
		                        & $100$ & $4.196$   & $3.879 \, \, (7.55\%)$  & $4.016 \, \, (4.29 \%)$ & $4.057 \, \, (3.31 \%)$ & $4.081 \, \, (2.73 \%)$ & \\ \midrule
		\multirow{6}{2em}{Put}  & $100$ & $4.469$   & $4.161 \, \, (6.89\%)$  & $4.306 \, \, (3.64 \%)$ & $4.354 \, \, (2.56 \%)$ & $4.396 \, \, (1.61 \%)$ & \\
		                        & $105$ & $7.171$   & $6.972 \, \, (2.77\%)$  & $7.081 \, \, (1.25 \%)$ & $7.125 \, \, (0.64 \%)$ & $7.178 \, \, (0.09 \%)$ & \\
		                        & $110$ & $10.86$   & $10.81 \, \, (0.44\%)$  & $10.85 \, \, (0.05 \%)$ & $10.87 \, \, (0.12 \%)$ & $10.91 \, \, (0.46 \%)$ & \\
		                        & $115$ & $15.38$   & $15.39 \, \, (0.06\%)$  & $15.38 \, \, (0.04 \%)$ & $15.39 \, \, (0.08 \%)$ & $15.40 \, \, (0.12 \%)$ & \\
		                        & $120$ & $20.30$   & $20.29 \, \, (0.08\%)$  & $20.29 \, \, (0.09 \%)$ & $20.29 \, \, (0.06 \%)$ & $20.31 \, \, (0.02 \%)$ & \\ \bottomrule
		                        & Time  &           & $9$s                    & $16$s                   & $24$s                   & $42$s                   & \\ \bottomrule
	\end{tabular}
	\caption[Pricing of European options in a Stationary Heston model with product hybrid recursive quantization with grids of size $(N_1,N_2)=(50,10)$.]{\textit{Comparison between European options prices, with maturity $T=0.5$ ($6$ months), given by quantization and the benchmark, in function of the strike $K$ and of the size $n$ where we set $(N_1,N_2)=(50,10)$.}}
	\label{RH:tab:price_european_fixedsize}
\end{table}

\paragraph{Bermudan options}
Then, in Figure \ref{RH:fig:US_callput100_N2_10}, we display the prices of monthly exercisable Bermudan options with maturity $T=0.5$ ($6$ months) for Call and Put of strikes $K=100$. The prices are computed by quantization and we compare the behavior of the pricer for different choices of time-step $n$ and sizes of the asset grids $N_1$ where we set $N_2=10$. Again, we notice that the choice of $n$ has a high impact on the price given by quantization compared to the choice of the grid size.

\begin{figure}[H]
	\centering
	\includegraphics[width=1.\textwidth]{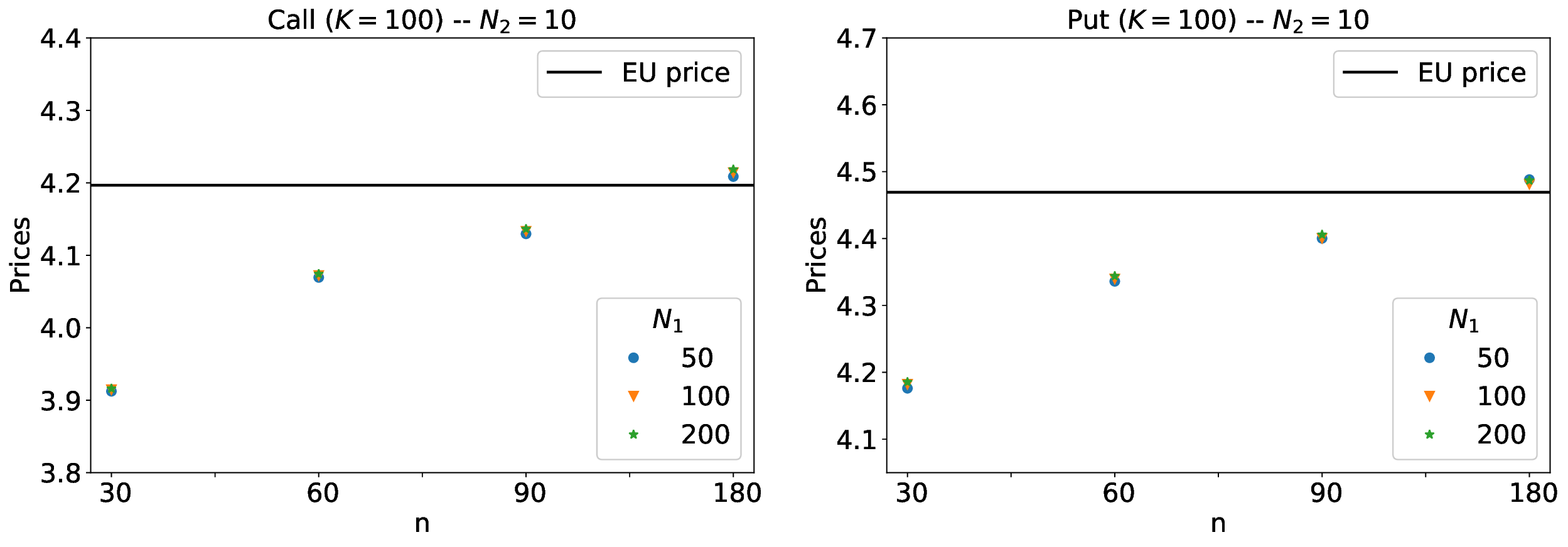}
	\caption[Prices of Bermudan options in the stationary Heston model given by product hybrid recursive quantization.]{\textit{Prices of Bermudan options in the stationary Heston model given by product hybrid recursive quantization with fixed value $N_2 = 10$.}}
	\label{RH:fig:US_callput100_N2_10}
\end{figure}

\paragraph{Barrier options}
Finally, in Figure \ref{RH:fig:price_barrier_L115}, we display the prices of an up-and-out Barrier option with strike $K=100$, maturity $T=0.5$ ($6$ months), barrier $L=115$ and $N_2 = 10$ computed with quantization. Again, we can notice the impact of $n$ on the approximated price.

\begin{figure}[H]
	\centering
	\includegraphics[width=.5\textwidth]{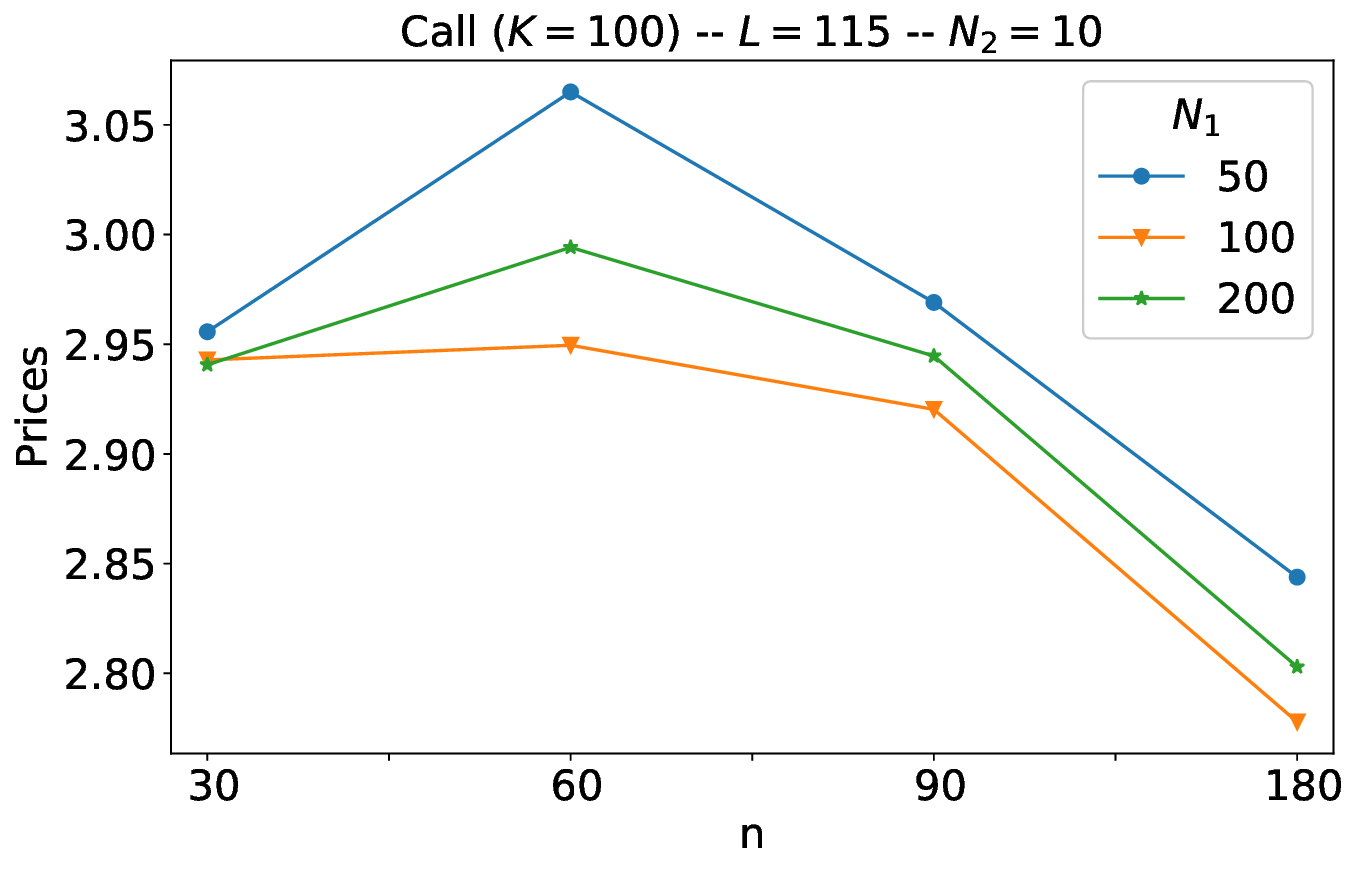}
	\caption[Prices of Barrier options with strike $K=100$ in the stationary Heston model given by product hybrid recursive quantization.]{\textit{Prices of Barrier options with strike $K=100$ in the stationary Heston model given by product hybrid recursive quantization with fixed value $N_2 = 10$.}}
	\label{RH:fig:price_barrier_L115}
\end{figure}

%% file: appendices.tex
\section{Discretization scheme for the volatility preserving the positivity} \label{RH:appendix:discussionschememilstein}

We recall the dynamics of the volatility
\begin{equation*}
	dv_t = \kappa (\theta - v_t) dt + \xi \sqrt{v_t} d \widetilde W_t
\end{equation*}
with $\kappa >0,\, \theta >0 \,\mbox{and}\, \xi > 0$. In this section, we discuss the choice of the discretization scheme under the Feller condition, which ensures the positivity of the process.

\paragraph{Euler-Maruyama scheme.}
Discretizing the volatility using an Euler-Maruyama scheme
\begin{equation*}
	\widebar v_{t_{k+1}} = \widebar v_{t_k} + \kappa (\theta - \widebar v_{t_k}) h + \xi \sqrt{\widebar v_{t_k}} \sqrt{h} \, Z_{k+1}^2
\end{equation*}
with $t_k = k h$, $h = T/n$ and $Z_{k+1}^2 = (\widetilde W_{t_{k+1}} - \widetilde W_{t_k} ) / \sqrt{h} $ may look natural. However, such a scheme clearly does not preserve positivity of the process even if the Feller condition is fulfilled since
\begin{equation*}
	\Prob \big( \widebar v_{t_1} < 0 \big) = \Prob \bigg( Z < \frac{- v_{0} - \kappa (\theta - v_0) h}{\xi \sqrt{v_0} \sqrt{h} } \bigg) > 0
\end{equation*}
with $Z \sim \N (0, 1)$. This suggests to introduce the Milstein scheme which is quite tractable in one dimension in absence of Lévy areas.

\paragraph{Milstein scheme.}
The Milstein scheme of the stochastic volatility is given by
\begin{equation*}
	\widebar{v}_{t_{k+1}} = \mathcal{M}_{b,\sigma} \big( t_k, \widebar{v}_{t_{k+1}}, Z_{k+1}^2 \big)
\end{equation*}
where (see \eqref{RH:MilsteinScheme})
\begin{equation*}
	\mathcal{M}_{b,\sigma} (t,x,z) = x - \frac{\sigma(x)}{2 \sigma_{x}^{\prime} (x)} + h \bigg(b(t, x) - \frac{ (\sigma \sigma_{x}^{\prime}) (x)}{2}\bigg) + \frac{ (\sigma \sigma_{x}^{\prime}) (x) h}{2} \bigg( z + \frac{1}{\sqrt{h} \sigma_{x}^{\prime} (x)} \bigg)^2.
\end{equation*}
with $b(x) = \kappa (\theta - x)$, $\sigma(x) = \xi \sqrt{x}$ and $ \sigma_{x}^{\prime} (x) = \frac{\xi}{2 \sqrt{x}}$. Consequently, under the Feller condition, the positivity of $\mathcal{M}_{b,\sigma} (t,x,z) $ is ensured if
\begin{equation*}
	x \geq \frac{\sigma(x)}{2 \sigma_{x}^{\prime} (x) } \geq 0, \qquad b(t, x) \geq \frac{ (\sigma \sigma_{x}^{\prime}) (x)}{2} \geq 0.
\end{equation*}
In our case, if the first condition holds true since
\begin{equation*}
	\frac{\sigma(x)}{2 \sigma_{x}^{\prime} (x) } = \frac{ \xi \sqrt{x} }{2 \frac{\xi}{2 \sqrt{x}} } = x
\end{equation*}
the second one fails. Indeed
\begin{equation*}
	\begin{aligned}
		\frac{ (\sigma \sigma_{x}^{\prime}) (x)}{2} = \frac{ \xi \sqrt{x} \frac{\xi}{2 \sqrt{x}} }{2} = \frac{ \xi^2 }{4}
	\end{aligned}
\end{equation*}
can be bigger than $b(t, x)$. In order to solve this problem, we consider the following \textit{boosted} volatility process
\begin{equation}
	Y_t = \e^{\kappa t} v_t, \, t \in [0,T].
\end{equation}

\paragraph{Milstein scheme for the \textit{boosted} volatility.}
Let $Y_t = \e^{\kappa t} v_t, \, t \in [0,T]$ for some $\kappa > 0$, which satisfies, owing to Itô's formula
\begin{equation*}
	d Y_t = \e^{\kappa t} \kappa \theta dt + \xi \e^{\kappa t / 2} \sqrt{Y_t} d \widetilde W_t.
\end{equation*}

\begin{remark}
	The process $(Y_t)_{t \in [0,T]}$ will have a higher variance but, having in mind a quantized scheme, this has no real impact (by contrast with a Monte Carlo simulation).
\end{remark}

Now, if we look at the Milstein discretization scheme of $Y_t$
\begin{equation*}
	\widebar{Y}_{t_{k+1}} = \mathcal{M}_{\widetilde b, \widetilde \sigma} \big( t_k, \widebar{Y}_{t_k}, Z_{k+1}^2 \big)
\end{equation*}
using the notation defined in \eqref{RH:MilsteinScheme} where drift and volatility terms of the \textit{boosted} process, now time-dependents, are given by
\begin{equation*}
	\widetilde b(t,x) = \e^{\kappa t} \kappa \theta, \qquad \widetilde \sigma(t,x) = \xi \sqrt{x} \e^{\kappa t / 2} \quad \textrm{ and } \quad \widetilde \sigma_{x}^{\prime} (t,x) = \frac{\xi \e^{\kappa t / 2}}{2 \sqrt{x}}.
\end{equation*}
Under the Feller condition, the positivity of the scheme is ensured, since
\begin{equation*}
	\frac{\widetilde \sigma(t,x)}{2 \widetilde \sigma_{x}^{\prime} (t,x) } = x \qquad \textrm{and} \qquad \frac{(\widetilde \sigma \widetilde \sigma_{x}^{\prime}) (t,x)}{2} = \frac{ \xi^2 \e^{ \kappa t}}{4} \leq \widetilde b(t,x) = \e^{\kappa t} \kappa \theta.
\end{equation*}
The last inequality is satisfied thanks to the condition $\frac{\xi^2}{2 \kappa \theta} \leq 1$ ensuring the positivity of the scheme.


\section{$L^p$-linear growth of the hybrid scheme}\label{RH:appendix:proofLplineargrowth}

The aim of this section is to show the $L^p$-linear growth of the scheme $F_k(u,z)$ with $u = (x,y)$ defined by
\begin{equation}\label{RH:scheme_couple}
	F_k(u, Z) = \Bigg(
	\begin{aligned}
			 & \mathcal{E}_{b, \sigma} \big( t_k, x, y, Z_{k+1}^1 \big)                    \\
			 & \mathcal{M}_{\widetilde b, \widetilde \sigma} \big( t_k, y, Z_{k+1}^2 \big)
		\end{aligned} \Bigg).
\end{equation}
where the schemes $\mathcal{E}_{b, \sigma}$ and $\mathcal{M}_{\widetilde b, \widetilde \sigma}$ are defined in \eqref{RH:eulerscheme} and \eqref{RH:MilsteinScheme}, respectively.

The results on the $L^p$-linear growth of the schemes are essentially based on the key Lemma 2.1 proved in \cite{sagna2018general} in $\R^d$ that we recall below.

\begin{lemme}\label{RH:key_lemma}
	\begin{enumerate}[label=(\alph*)]
		\item Let $u \in \R^d$ and $A(u)$ be a $d \times q$-matrix and let $a(u) \in \R^d$. Let $p \in [2,3)$. For any centered random vector $\zeta \in L^p_{\R^d}(\Omega, \A, \Prob)$, one has for every $h \in (0, + \infty)$
		      \begin{equation}
			      \E \big[ \vert a(u) + \sqrt{h} A(u) \zeta \vert^p \big] \leq \bigg( 1 + \frac{(p-1)(p-2)}{2} h \bigg) \vert a(u) \vert^p + h \big( 1 + p + h^{\frac{p}{2}-1} \big) \Vert A(u) \Vert^p \E \big[ \vert \zeta \vert^p \big]
		      \end{equation}
		      where $\Vert A(u) \Vert = \big( \Tr(A(u) A^{\star}(u)) \big)^{1/2}$.
		\item In particular, if $\vert a(u) \vert \leq \vert u \vert ( 1 + L h) +  Lh$ and $\Vert A(u) \Vert^p \leq 2^{p-1} \Upsilon^p (1 + \vert u \vert^p)$, then
		      \begin{equation}
			      \E \big[ \vert a(u) + \sqrt{h} A(u) \zeta \vert^p \big] \leq \big( \e^{\kappa_p h} L + K_p \big) h + \big( \e^{\kappa_p h} + K_p h \big) \vert u \vert^p,
		      \end{equation}
		      where
		      \begin{equation}
			      \kappa_p = \frac{(p-1)(p-2)}{2} + 2 p L \quad \mbox{and} \quad K_p = 2^{p-1} \Upsilon^p \big( 1 + p + h^{\frac{p}{2}-1} \big) \E \big[ \vert \zeta \vert^p \big].
		      \end{equation}
	\end{enumerate}
\end{lemme}

Now, we will apply Lemma \ref{RH:key_lemma} to $F_k(u,z)$ defined in \eqref{RH:scheme_couple} further on in order to show its $L^p$-linear growth. Let $a(u) \in \R^2$ and let $A(u)$ be a $2 \times 3$-matrix defined by
\begin{equation*}
	\begin{aligned}
		a(u) = &
		\begin{pmatrix}
			x + h \big(r - \frac{\e^{- \kappa t_k} y}{2} \big) \\
			y + \e^{\kappa t_k} \kappa \theta h                \\
		\end{pmatrix},
		\quad A(u) =
		\begin{pmatrix}
			\e^{-\kappa t_k / 2} \sqrt{y} & 0                            & 0                                        \\
			0                             & \sqrt{y} \e^{\kappa t_k / 2} & \sqrt{h} \frac{\xi^2 \e^{\kappa t_k}}{4} \\
		\end{pmatrix}                             \\
		       & \qquad \qquad \qquad \mbox{and} \quad \zeta =
		\begin{pmatrix}
			Z_{k+1}^1         \\
			Z_{k+1}^2         \\
			(Z_{k+1}^2)^2 - 1 \\
		\end{pmatrix}.
	\end{aligned}
\end{equation*}

First, we show the linear growth of $a(u)$
\begin{equation*}
	\begin{aligned}
		\vert a(u) \vert
		 & = \Big( \Big\vert x + h \big(r - \frac{\e^{- \kappa t_k} y}{2} \big) \Big\vert^2 + \big\vert y + \e^{\kappa t_k} \kappa \theta h \big\vert^2 \Big)^{1/2}              \\
		 & = \Big( \vert x \vert^2 + \vert y \vert^2 + h^2 \Big( r^2 + \frac{\e^{- 2 \kappa t_k}}{4} \vert y \vert^2 \Big) + \e^{2 \kappa t_k} \kappa^2 \theta^2 h^2 \Big)^{1/2} \\
		 & \leq \Big( \vert u \vert^2 \Big( 1 + h^2 \frac{\e^{- 2 \kappa t_k}}{4} \Big) + h^2 \big( r^2 + \e^{2 \kappa t_k} \kappa^2 \theta^2 \big) \Big)^{1/2}                  \\
		 & \leq \vert u \vert \Big( 1 + h^2 \frac{\e^{- 2 \kappa t_k}}{4} \Big)^{1/2} + h \big( r^2 + \e^{2 \kappa t_k} \kappa^2 \theta^2 \big)^{1/2}                            \\
		 & \leq \vert u \vert \Big( 1 + h \frac{h}{2} \Big) + h \big( r^2 + \e^{2 \kappa T} \kappa^2 \theta^2 \big)^{1/2}                                                        \\
		 & \leq \vert u \vert ( 1 + L h ) + L h
	\end{aligned}
\end{equation*}
where $L = \max \Big( \frac{1}{2} , \big( r^2 + \e^{2 \kappa T} \kappa^2 \theta^2 \big)^{1/2} \Big)$. Then, we study $\Vert A(u) \Vert^p$
\begin{equation*}
	\begin{aligned}
		\Vert A(u) \Vert^p
		 & = \Big( \e^{-\kappa t_k} \vert y \vert + \vert y \vert \e^{\kappa t_k} + h \frac{\xi^4 \e^{2 \kappa t_k}}{16} \Big)^{p/2}                                                                                       \\
		 & = \Big( \vert y \vert ( \e^{-\kappa t_k} + \e^{\kappa t_k} ) + h \frac{\xi^4 \e^{2 \kappa t_k}}{16} \Big)^{p/2}                                                                                                 \\
		 & \leq 2^{\frac{p}{2}-1} \Big( \vert y \vert^{\frac{p}{2}} ( \e^{-\kappa t_k} + \e^{\kappa t_k} )^{\frac{p}{2}} + h^{\frac{p}{2}} \frac{\xi^{2p} \e^{p \kappa t_k}}{4^p} \Big)                                    \\
		 & \leq 2^{\frac{p}{2}-1} \Big( \frac{\vert y \vert^{p} + 1}{2} ( \e^{-\kappa t_k} + \e^{\kappa t_k} )^{\frac{p}{2}} + h^{\frac{p}{2}} \frac{\xi^{2p} \e^{p \kappa t_k}}{4^p} \Big)                                \\
		 & \leq 2^{\frac{p}{2}-1} \frac{( 1 + \e^{\kappa T} )^{\frac{p}{2}}}{2} \Big( \vert y \vert^{p} + 1 + h^{\frac{p}{2}} \frac{\xi^{2p} \e^{p \kappa T}}{2^{2p-1}} \frac{1}{( 1 + \e^{\kappa T} )^{\frac{p}{2}}}\Big) \\
		 & \leq 2^{p-1} \Upsilon^p \big( 1 + \vert u \vert^{p} \big)
	\end{aligned}
\end{equation*}
where $\Upsilon^p = \frac{( 1 + \e^{\kappa T} )^{\frac{p}{2}}}{2} + h^{\frac{p}{2}} \frac{\xi^{2p} \e^{p \kappa T}}{2^{2p}}$. Hence, by Lemma \ref{RH:key_lemma}, the discretization scheme $F_k$ has an $L^p$-linear growth
\begin{equation*}
	\E \big[ \vert F_k(u, Z_{k+1}) \vert^p \big] \leq \alpha_{p} + \beta_{p} \vert u \vert^p
\end{equation*}
with
\begin{equation} \label{RH:coefficients_sub_linear}
	\alpha_{p} = \big( \e^{\kappa_p h} L + K_p \big) h \quad \mbox{and} \quad \beta_{p} = \e^{\kappa_p h} + K_p h
\end{equation}
where $K_p$ and $\kappa_p$ are defined in the Lemma \ref{RH:key_lemma}.

\section{Proof of the $L^2$-error estimation of Proposition \ref{RH:prop:l2-error}} \label{RH:appendix:proofl2error}

We have, for every $k=0, \dots, n-1$
\begin{equation}
	\begin{aligned}
		\widehat U_{k+1} - \widebar U_{k+1}
		 & = \widehat U_{k+1} - \widetilde U_{k+1} + \widetilde U_{k+1} - \widebar U_{k+1}                         \\
		 & = \widehat U_{k+1} - \widetilde U_{k+1} + F_k ( \widehat U_k, Z_{k+1} ) - F_k ( \widebar U_k, Z_{k+1} ) \\
	\end{aligned}
\end{equation}
by the very definition of $\widetilde U_{k+1}$ and $\widebar U_{k+1}$. Hence,
\begin{equation} \label{RH:L2_error_recurquantif_couple}
	\begin{aligned}
		\Vert \widehat U_{k+1} - \widebar U_{k+1} \Vert_{_2}
		 & \leq \Vert \widehat U_{k+1} - \widetilde U_{k+1} \Vert_{_2} + \Vert \widetilde U_{k+1} - \widebar U_{k+1} \Vert_{_2}                    \\
		 & \leq \Vert \widehat U_{k+1} - \widetilde U_{k+1} \Vert_{_2} + \Vert F_k(\widehat U_k, Z_{k+1}) - F_k(\widebar U_k, Z_{k+1}) \Vert_{_2}.
	\end{aligned}
\end{equation}
Using the definition of Milstein scheme of the \textit{boosted}-volatility models $\mathcal{M}_{\widetilde b, \widetilde \sigma}$ in \eqref{RH:MilsteinScheme_with_specificationmodel}, the $\frac{1}{2}$-Hölder property of $\sqrt{x}$, for every $y,y' \in \R_+$ one has
\begin{equation}
	\begin{aligned}
		\big\vert \mathcal{M}_{\widetilde b, \widetilde \sigma} \big( t, y, z \big) - \mathcal{M}_{\widetilde b, \widetilde \sigma} \big( t, y', z \big) \big\vert
		 & = \bigg\vert \Big( z \frac{ \xi \e^{\kappa t / 2} \sqrt{h} }{2} + \sqrt{y} \Big)^2 - \Big( z \frac{ \xi \e^{\kappa t / 2} \sqrt{h} }{2} + \sqrt{y'} \Big)^2 \bigg\vert \\
		 & \leq \big\vert \sqrt{y} - \sqrt{y'} \big\vert \big( \vert z \vert \xi \e^{\kappa t / 2} \sqrt{h} + \sqrt{y} + \sqrt{y'} \big)                                          \\
		 & \leq \sqrt{ \vert y - y' \vert} \sqrt{h} \vert z \vert \xi \e^{\kappa t / 2}  + \vert y - y' \vert
	\end{aligned}
\end{equation}
and using the definition of the Euler-Maruyama scheme of the log-asset $\mathcal{E}_{b, \sigma}$ defined in \eqref{RH:eulerscheme} we have, for any $x,x',y,y' \in \R_+$
\begin{equation}
	\begin{aligned}
		\big\vert \mathcal{E}_{b, \sigma} \big( t, x, y, z \big) - \mathcal{E}_{b, \sigma} \big( t, x', y', z \big) \big\vert \leq \vert x - x' \vert + \frac{\e^{- \kappa t}}{2} h \vert y - y' \vert + \e^{- \kappa t / 2} \sqrt{h} \vert z \vert \sqrt{ \vert y - y' \vert}.
	\end{aligned}
\end{equation}

Now, when we replace $x,y,x',y'$ by $\widehat X_k, \widehat Y_k, \widebar X_k, \widebar Y_k$ in the last expression, we get an upper-bound for the last term of \eqref{RH:L2_error_recurquantif_couple}
\begin{equation}
	\begin{aligned}
		 & \big\Vert F_k(\widehat U_k, Z_{k+1}) - F_k(\widebar U_k, Z_{k+1}) \big\Vert_{_2}                                                                                                                                                                                                                       \\
		 & \qquad \leq \big\Vert \mathcal{E}_{b, \sigma} \big( t_k, \widehat X_k, \widehat Y_k, Z_{k+1}^1 \big) - \mathcal{E}_{b, \sigma} \big( t_k, \widebar X_k, \widebar Y_k, Z_{k+1}^1 \big) \big\Vert_{_2}                                                                                                   \\
		 & \qquad \qquad \qquad \qquad + \big\Vert \mathcal{M}_{\widetilde b, \widetilde \sigma} \big( t_k, \widehat Y_k, Z_{k+1}^2 \big) - \mathcal{M}_{\widetilde b, \widetilde \sigma} \big( t_k, \widebar Y_k, Z_{k+1}^2 \big) \big\Vert_{_2}                                                                 \\
		 & \qquad \leq \Vert \widehat X_k - \widebar X_k \Vert_{_2} + \Big( 1 + \frac{ \e^{- \kappa t_k}}{2} h \Big) \Vert \widehat Y_k - \widebar Y_k \Vert_{_2} + \Big\Vert \sqrt{h} \big( \xi \e^{\kappa t_k /2 } + \e^{- \kappa t_k /2} \big) \sqrt{ \vert \widehat Y_k - \widebar Y_k \vert } \Big\Vert_{_2} \\
		 & \qquad \leq \Vert \widehat X_k - \widebar X_k \Vert_{_2} + \Big( 1 + \frac{ \e^{- \kappa t_k}}{2} h \Big) \Vert \widehat Y_k - \widebar Y_k \Vert_{_2} + \Big\Vert \sqrt{2h ( \xi^2 \e^{\kappa t_k} + \e^{- \kappa t_k} )} \sqrt{ \vert \widehat Y_k - \widebar Y_k \vert } \Big\Vert_{_2}.
	\end{aligned}
\end{equation}
Now, using that $\sqrt{a} \sqrt{b} \leq \frac{1}{2} \big( \frac{a}{\lambda} + b \lambda \big)$ with $\sqrt{a} = \sqrt{2h ( \xi^2 \e^{\kappa t_k} + \e^{- \kappa t_k} )}$ and $\sqrt{b} = \sqrt{ \vert \widehat Y_k - \widebar Y_k \vert }$ where we considere that $\lambda = \sqrt{h} (1 - \sqrt{h})$. Wo choose $\lambda$ of this order because we wish to divide equally the impact of $h$ and get $\sqrt{h}$ on each side. Hence, we have
\begin{equation}
	\begin{aligned}
		\sqrt{2h ( \xi^2 \e^{\kappa t_k} + \e^{- \kappa t_k} )} \sqrt{ \vert \widehat Y_k - \widebar Y_k \vert } \leq \frac{1}{2} \bigg( \frac{2h ( \xi^2 \e^{\kappa t_k} + \e^{- \kappa t_k} ) }{ \lambda } + \vert \widehat Y_k - \widebar Y_k \vert \lambda \bigg).
	\end{aligned}
\end{equation}
Then,
\begin{equation}
	\begin{aligned}
		 & \big\Vert F_k(\widehat U_k, Z_{k+1}) - F_k(\widebar U_k, Z_{k+1}) \big\Vert_{_2}                                                                                                                                                                                                                                                  \\
		 & \qquad \leq \Vert \widehat X_k - \widebar X_k \Vert_{_2} + \Big( 1 + \frac{ \e^{- \kappa t_k}}{2} h \Big) \Vert \widehat Y_k - \widebar Y_k \Vert_{_2} + \Big\Vert \frac{1}{2} \bigg( \frac{2h ( \xi^2 \e^{\kappa t_k} + \e^{- \kappa t_k} ) }{ \lambda } + \vert \widehat Y_k - \widebar Y_k \vert \lambda \bigg) \Big\Vert_{_2} \\
		 & \qquad \leq \Vert \widehat X_k - \widebar X_k \Vert_{_2} + \Big( 1 + \frac{ \e^{- \kappa t_k}}{2} h + \frac{\lambda}{2} \Big) \Vert \widehat Y_k - \widebar Y_k \Vert_{_2} + ( \xi^2 \e^{\kappa t_k} + \e^{- \kappa t_k} ) \frac{h}{ \lambda }                                                                                    \\
		 & \qquad \leq \sqrt{2} \Big( 1 + \frac{h}{2} + \frac{\lambda}{2} \Big) \Vert \widehat U_k - \widebar U_k \Vert_{_2} + C_{T} \frac{h}{ \lambda }                                                                                                                                                                                     \\
		 & \qquad \leq \sqrt{2} \Big( 1 + \frac{\sqrt{h}}{2} \Big) \Vert \widehat U_k - \widebar U_k \Vert_{_2} + C_{T} (h) \sqrt{h}
	\end{aligned}
\end{equation}
where $C_T (h) = (1 + \xi^2 \e^{\kappa T}) (1-\sqrt{h})^{-1} = O(1)$.

Finally, \eqref{RH:L2_error_recurquantif_couple} is upper-bounded by
\begin{equation}
	\begin{aligned}
		\Vert \widehat U_{k+1} - \widebar U_{k+1} \Vert_{_2}
		 & \leq \Vert \widehat U_{k+1} - \widetilde U_{k+1} \Vert_{_2} + \sqrt{2} \Big( 1 + \frac{\sqrt{h}}{2} \Big) \Vert \widehat U_k - \widebar U_k \Vert_{_2} + C_{T} (h) \sqrt{h}                                                        \\
		 & \leq \sum_{j=0}^{k+1} \Vert \widehat U_j - \widetilde U_j \Vert_{_2} 2^{\frac{k-j+1}{2}} \Big( 1 + \frac{\sqrt{h}}{2} \Big)^{k-j+1} + \sqrt{h} C_{T} (h) \sum_{j=0}^{k} 2^{\frac{k-j}{2}} \Big( 1 + \frac{\sqrt{h}}{2} \Big)^{k-j} \\
		 & \leq \sum_{j=0}^{k+1} A_{j,k+1} \Vert \widehat U_j - \widetilde U_j \Vert_{_2} + B_{k+1} \sqrt{h}                                                                                                                                  \\
	\end{aligned}
\end{equation}
where
\begin{equation}
	A_{j,k} = 2^{\frac{k-j}{2}} \e^{\frac{\sqrt{h}}{2}(k-j)} \quad \mbox{and} \quad B_k = C_{T} (h) \sum_{j=0}^{k-1} 2^{\frac{k-1-j}{2}} \e^{\frac{\sqrt{h}}{2}(k-1-j)}
\end{equation}
and $\sum_{\emptyset} = 0$ by convention.

\medskip
Now, we follow the lines of the proof developed in \cite{sagna2018general}, we apply the revisited Pierce's lemma for product quantization (Lemma 2.3 in \cite{sagna2018general}) with $r=2$ and let $p>r=2$, which yields
\begin{equation} \label{RH:step_proof_error_l2}
	\Vert \widehat U_{k+1} - \widebar U_{k+1} \Vert_{_2} \leq 2^{\frac{p-2}{2p}} C_p  \sum_{j=0}^{k+1} A_{j,k+1} \Vert \widetilde U_j \Vert_{_p} \big( N_{1,j} \times N_{2,j} \big)^{-1/2} + B_{k+1} \sqrt{h}
\end{equation}
where $C_p = 2 C_{1,p}$ and $C_{1,p}$ is the constant appearing in Pierce lemma (see the second item in Theorem \ref{RH:zador} and \cite{graf2000foundations} for further details) and we used that $\Vert \widetilde U_j \Vert_{_p} \geq \sigma_p (\widetilde U_j) = \inf_{a \in \R^2} \Vert \widetilde U_j - a \Vert_{_p}$.
Moreover, noting that the hybrid discretization scheme $F_k$ has an $L^p$-linear growth, (see Appendix \ref{RH:appendix:proofLplineargrowth}), i.e.
\begin{equation}
	\forall k = 0, \dots, n-1, \quad \forall u \in \R^2, \quad \E \big[ \vert F_k(u, Z_{k+1}) \vert^p \big] \leq \alpha_{p} + \beta_{p} \vert x \vert^p,
\end{equation}
where the coefficients $\alpha_{p}$ and $\beta_{p}$ are defined in \eqref{RH:coefficients_sub_linear}. Hence, for all $j = 0, \dots, n-1$, we have
\begin{equation}\label{RH:step_proof_norm_l2_of_widetildeU}
	\Vert \widetilde U_{j+1} \Vert_{_p}^p = \E \big[ \E \big[ \vert F_j ( \widehat U_j, Z_{j+1} ) \vert^p \mid \widehat U_j \big] \big] \leq \alpha_{p} + \beta_{p} \Vert \widehat U_j \Vert_{_p}^p.
\end{equation}
Furthermore, $\E \big[ \vert \widehat U_j \vert^p \big]$ can be upper-bounded using Jensen's inequality and the stationary property satisfied by $\widehat X_j$ and $\widehat Y_j$ independently. Indeed, they are one-dimensional quadratic optimal quantizers of $\widetilde X_j$ and $\widetilde Y_j$, respectively, hence they are stationary in the sense of Proposition \ref{RH:stationary_property}.
\begin{equation}
	\begin{aligned}
		\Vert \widehat U_j \Vert_{_p}^p
		 & \leq 2^{\frac{p}{2}-1} \Big( \E \big[ \vert \widehat X_j \vert^p \big] + \E \big[ \vert \widehat Y_j \vert^p \big] \Big)                                                                                         \\
		 & \leq 2^{\frac{p}{2}-1} \bigg( \E \Big[ \big\vert \E \big[ \widetilde X_j \mid \widehat X_j \big] \big\vert^p \Big] + \E \Big[ \big\vert \E \big[ \widetilde Y_j \mid \widehat Y_j \big] \big\vert^p \Big] \bigg) \\
		 & \leq 2^{\frac{p}{2}-1} \Big( \E \big[ \vert \widetilde X_j \vert^p \big] + \E \big[ \vert \widetilde Y_j \vert^p \big] \Big)                                                                                     \\
		 & = 2^{\frac{p}{2}-1} \Vert \widetilde U_j \Vert_{_p}^p                                                                                                                                                            \\
		 & \leq 2^{\frac{p}{2}-1} \Vert \widetilde U_j \Vert_{_2}^p.
	\end{aligned}
\end{equation}
Now, plugging this upper-bound in \eqref{RH:step_proof_norm_l2_of_widetildeU} and by a standard induction argument, we have
\begin{equation}\label{RH:upperbound_widetildeU}
	\begin{aligned}
		\Vert \widetilde U_j \Vert_{_p}^p
		 & \leq \alpha_{p} + \beta_{p} 2^{\frac{p}{2}-1} \Vert \widetilde U_{j-1} \Vert_{_2}^p                                                                                  \\
		 & \leq 2^{(\frac{p}{2}-1 )j} \beta_{p}^j \Vert \widehat U_0 \Vert_{_2}^p + \alpha_{p} \sum_{i=0}^{j-1} \big( 2^{\frac{p}{2}-1} \beta_{p} \big)^{i}                     \\
		 & \leq 2^{(\frac{p}{2}-1 )j} \beta_{p}^j \Vert \widehat U_0 \Vert_{_2}^p + \alpha_{p} \frac{1 - 2^{(\frac{p}{2}-1 )j} \beta_{p}^j }{1 - 2^{\frac{p}{2}-1} \beta_{p} }.
	\end{aligned}
\end{equation}
Hence, using the upper-bound \eqref{RH:upperbound_widetildeU} in \eqref{RH:step_proof_error_l2}, we have
\begin{equation}
	\begin{aligned}
		\Vert \widehat U_{k+1} & - \widebar U_{k+1} \Vert_{_2}                                                                                                                                                                                                                                                                            \\
		                       & \leq 2^{\frac{p-2}{2p}} C_p^2 \sum_{j=0}^{k+1} A_{j,k+1} \bigg( 2^{(\frac{p}{2}-1 )j} \beta_{p}^j \Vert \widehat U_0 \Vert_{_2}^p + \alpha_{p} \frac{1 - 2^{(\frac{p}{2}-1 )j} \beta_{p}^j }{1 - 2^{\frac{p}{2}-1} \beta_{p} } \bigg)^{1/p} \big( N_{1,j} \times N_{2,j} \big)^{-1/2} + B_{k+1} \sqrt{h} \\
		                       & \leq \sum_{j=0}^{k+1} \widetilde A_{j,k+1} \big( N_{1,j} \times N_{2,j} \big)^{-1/2} + B_{k+1} \sqrt{h}                                                                                                                                                                                                  \\
	\end{aligned}
\end{equation}
yielding the desired result with
\begin{equation}
	\widetilde A_{j,k} = 2^{\frac{p-2}{2p}} C_p^2 A_{j,k} \bigg( 2^{(\frac{p}{2}-1 )j} \beta_{p}^j \Vert \widehat U_0 \Vert_{_2}^p + \alpha_{p} \frac{1 - 2^{(\frac{p}{2}-1 )j} \beta_{p}^j }{1 - 2^{\frac{p}{2}-1} \beta_{p} } \bigg)^{1/p}.
\end{equation}


\section{Quadratic Optimal Quantization: Generic Approach} \label{RH:appendix:optquant}

Let $X$ be a $\R$-valued random variable with distribution $\Prob_{_{X}}$ defined on a probability space $ ( \Omega, \A, \Prob )$ such that $X \in L^2_{\R} ( \Omega, \A, \Prob )$.

\begin{definition}
	Let $\Gamma_N = \{ x_1^N, \dots, x_N^N \} \subset \R$ be a subset of size $N$, called $N$-quantizer. A Borel partition $( C_i (\Gamma_N))_{i \in \{ 1, \dots, N \} }$ of $\R$ is a Voronoï partition of $\R$ induced by the $N$-quantizer $\Gamma_N$ if, for every $i \in \{ 1, \dots, N \} $,
	\begin{equation*}
		C_i (\Gamma_N) \subset \big\{ \xi \in \R, \vert \xi - x_i^N \vert \leq \min_{j \neq i }\vert \xi - x_j^N \vert \big\}.
	\end{equation*}
	The Borel sets $C_i (\Gamma_N)$ are called Voronoï cells of the partition induced by $\Gamma_N$.
\end{definition}

\begin{remark}
	Any such $N$-quantizer is in correspondence with the $N$-tuple $x = ( x_1^N, \dots, x_N^N ) \in ( \R )^N$ as well as with all $N$-tuples obtained by a permutation of the components of $x$. This is why we will sometimes replace $\Gamma_N$ by $x$.
\end{remark}

\medskip
If the quantizers are in non-decreasing order: $x_1^N < x_2^N < \cdots < x_{N-1}^N < x_{N}^N$, then the Voronoï cells are given by
\begin{equation}\label{RH:def_voronoi_cells}
	C_i ( \Gamma_N ) = \big( x_{i - 1/2}^N, x_{i + 1/2}^N \big], \qquad i \in \{ 1, \dots,  N-1 \}, \qquad C_{N} ( \Gamma_N ) = \big( x_{N - 1/2}^N, x_{N + 1/2}^N \big)
\end{equation}
where $\forall i \in \in \{ 2, \dots,  N \}, x_{i-1/2}^N = \frac{x_{i-1}^N + x_i^N}{2}$ and $x_{1/2}^N = - \infty$ and $x_{N+1/2}^N =  + \infty$.

\begin{definition}
	The Vorono\"i quantization of $X$ by $\Gamma_N$, $\widehat X^N$, is defined as the nearest neighbour projection of $X$ onto $\Gamma_N$
	\begin{equation}\label{RH:quantizOfX}
		\widehat X^N = \Proj_{\Gamma_N} (X) = \sum_{i = 1}^N x_i^N \1_{X \in C_i (\Gamma_N) }
	\end{equation}
	and its associated probabilities, also called weights, are given by
	\begin{equation*}
		\Prob \big( \widehat X^N = x_i^N \big) = \Prob_{_{X}} \big( C_i (\Gamma_N) \big) = \Prob \Big( X \in \big( x_{i - 1/2}^N, x_{i + 1/2}^N \big] \Big).
	\end{equation*}
\end{definition}

\begin{definition}
	The quadratic distortion function at level $N$ induced by an $N$-tuple $x = (x_1^N, \dots, x_N^N) $ is given by
	\begin{equation*}
		\Distortion : x \longmapsto \frac{1}{2} \E \Big[ \min_{i \in \{ 1, \dots, N \} } \vert X - x_i^N \vert^2 \Big] = \frac{1}{2} \E \big[ \dist ( X, \Gamma_N )^2 \big] = \frac{1}{2} \Vert X - \widehat X^N \Vert_{_2}^2 .
	\end{equation*}
\end{definition}

Of course, the above result can be extended to the $L^p$ case by considering the $L^p$-mean quantization error in place of the quadratic one.

We briefly recall some classical theoretical results, see \cite{graf2000foundations,pages2018numerical} for further details. The first one treats of existence of optimal quantizers.
\begin{theorem}{(Existence of optimal $N$-quantizers)}\label{RH:existence}
	Let $X \in L^2_{\R} ( \Prob )$ and $N \in \Integer^{\star}$.
	\begin{enumerate}[label=(\alph*)]
		\item The quadratic distortion function $\Distortion$ at level $N$ attains a minimum at a $N$-tuple $x^{\star} = ( x_1^N, \dots, x_N^N )$ and $\Gamma_N^{\star} = \big\{ x_i^{N}, \, i \in \{ 1, \dots, N \}  \big\}$ is a quadratic optimal quantizer at level $N$.
		\item If the support of the distribution $\Prob_{_{X}}$ of $X$ has at least $N$ elements, then $x^{\star} = ( x_1^N, \dots, x_N^N )$ has pairwise distinct components, $ \Prob_{_{X}} \big( C_i ( \Gamma_N^{\star} ) \big) > 0, \, i \in \{ 1, \dots, N \} $. Furthermore, the sequence $N \mapsto \inf_{x \in ( \R )^N} \Distortion( x )$ converges to $0$ and is decreasing as long as it is positive.
	\end{enumerate}
\end{theorem}

A really interesting and useful property concerning quadratic optimal quantizers is the stationary property, this property is deeply connected to the addressed problem after for the optimization of the quadratic optimal quantizers in \eqref{RH:distortozero}.

\begin{proposition}{(Stationarity)}\label{RH:stationary_property}
	Assume that the support of $\Prob_{_{X}}$ has at least $N$ elements. Any $L^2$-optimal $N$-quantizer $\Gamma_N \in ( \R )^N$ is stationary in the following sense: for every Vorono\"{i} quantization $\widehat X^N$ of $X$,
	\begin{equation*}
		\E \big[ X \mid \widehat X^N \big] = \widehat X^N.
	\end{equation*}
	Moreover $\Prob \big( X \in \bigcup_{i = 1, \dots, N} \partial C_i (\Gamma_N) \big) = 0$, so all optimal quantization induced by $\Gamma_N$ a.s. coincide.
\end{proposition}

The uniqueness of an optimal $N$-quantizer, due to Kieffer \cite{kieffer1982exponential}, was shown in dimension one under some assumptions on the density of $X$.
\begin{theorem}{(Uniqueness of optimal $N$-quantizers see \cite{kieffer1982exponential})}\label{RH:uniqueness}
	If $\Prob_{_{X}}(d\xi) = \varphi (\xi) d \xi$ with $\log \varphi$ concave, then for every $N \geq 1$, there is exactly one stationary $N$-quantizer (up to the permutations of the $N$-tuple). This unique stationary quantizer is a global (local) minimum of the distortion function, i.e.
	\begin{equation*}
		\forall N \geq 1, \qquad \argmin_{\R^N} \Distortion = \{ x^{\star} \}.
	\end{equation*}
\end{theorem}
In what follows, we will drop the star notation ($\star$) when speaking of optimal quantizers, $x^{\star}$ and $\Gamma_N^{\star}$ will be replaced by $x$ and $\Gamma_N$.

The next result elucidates the asymptotic behavior of the distortion. We saw in Theorem \ref{RH:existence} that the infimum of the quadratic distortion converges to $0$ as $N$ goes to infinity. The next theorem, known as Zador's Theorem, establishes the sharp rate of convergence of the $L^p$-mean quantization error.

\begin{theorem}{(Zador's Theorem)}\label{RH:zador} Let $p \in (0, + \infty)$.
	\begin{enumerate}[label=(\alph*)]
		\item {\sc Sharp rate \cite{zador1982asymptotic,graf2000foundations}}. Let $X \in L^{p+ \delta}_{\R}(\Prob)$ for some $\delta > 0$. Let $\Prob_{_{X}} (d \xi) = \varphi(\xi) \cdot \lambda ( d \xi ) + \nu ( d \xi ) $, where $\nu ~ \bot ~ \lambda$ i.e., is singular with respect to the Lebesgue measure $\lambda$ on $\R$. Then, there is a constant $\widetilde{J}_{p,1} \in (0, + \infty)$ such that
		      \begin{equation}
			      \lim_{N \rightarrow + \infty} N \min_{\Gamma_N \subset \R, \vert \Gamma_N \vert \leq N } \Vert X - \widehat{X}^N \Vert_{_p} = \frac{1}{2^p (p+1)} \bigg[ \int_{\R} \varphi^{\frac{1}{1+p}} d \lambda \bigg]^{1+\frac{1}{p} }.
		      \end{equation}

		\item {\sc Non asymptotic upper-bound \cite{graf2000foundations,pages2018numerical}}. Let $\delta > 0$. There exists a real constant $C_{1,p} \in (0, +\infty )$ such that, for every $\R$-valued random variable $X$,
		      \begin{equation}
			      \forall N \geq 1, \qquad \min_{\Gamma_N \subset \R, \vert \Gamma_N \vert \leq N } \Vert X - \widehat{X}^N \Vert_{_p} \leq C_{1,p} \sigma_{\delta+p} (X) N^{- 1}
		      \end{equation}
		      where, for $r \in (0, + \infty),\sigma_r(X) = \min_{a \in \R} \left\Vert X - a \right\Vert_{_r} < + \infty$ is the $L^r$-pseudo-standard deviation.
	\end{enumerate}
\end{theorem}

Now, we will be interested by the construction of such quadratic optimal quantizer. We differentiate $\Distortion$, whose gradient is given by
\begin{equation}
	\nabla \Distortion ( x )
	= \bigg( \E \Big[ ( x_i^N - X ) \1_{ X \, \in \, \big( x_{i - 1/2}^N, x_{i + 1/2}^N \big] } \Big] \bigg)_{i = 1, \dots, N}.
\end{equation}
Moreover, if $x$ is solution to the distortion minimization problem then it satisfies
\begin{equation}\label{RH:distortozero}
	\begin{aligned}
		\nabla \Distortion ( x ) = 0
		\quad \iff \quad & x_i^N = \frac{ \E \Big[ X \1_{ X \, \in \, \big( x_{i - 1/2}^N, x_{i + 1/2}^N \big] } \Big] }{ \Prob \Big( X \in \big( x_{i - 1/2}^N, x_{i + 1/2}^N \big] \Big) }, \qquad i = 1, \dots, N \\
		\quad \iff \quad & x_i^N = \frac{ K_{_X} \big( x_{i + 1/2}^N \big) - K_{_X} \big( x_{i - 1/2}^N \big) }{ F_{_X} \big( x_{i + 1/2}^N \big) - F_{_X} \big( x_{i - 1/2}^N \big) }, \qquad i = 1, \dots, N       \\
	\end{aligned}
\end{equation}
where $K_{_X}(\cdot)$ and $F_{_X}(\cdot)$ are the first partial moment and the cumulative distribution respectively, function of $X$, i.e.
\begin{equation}
	K_{_X}( x ) = \E \big[ X \1_{X \leq x} \big] \qquad \textrm{ and } \qquad F_{_X}( x ) = \Prob \big( X \leq x \big).
\end{equation}

Hence, one can notices that the optimal quantizer that cancel the gradient defined in \eqref{RH:distortozero}, hence is an optimal quantizer, is a stationary quantizer in the following sense
\begin{equation}
	\E \big[ \widehat x^N \mid X \big] = \widehat X^N.
\end{equation}

The last equality in \eqref{RH:distortozero} was the starting point to the development of the first method devoted to the numerical computation of optimal quantizers: the Lloyd's method I. This method was first devised in 1957 by S.P. Lloyd and published later \cite{lloyd1982least}. Starting from a sorted $N$-tuple $x^{[0]}$ and with the knowledge of the first partial moment $K_{_X}$ and the cumulative distribution function $F_{_X}$ of $X$, the algorithm, which is essentially a deterministic fixed point method, is defined as follows
\begin{equation}\label{RH:LloydAlgo}
	\begin{aligned}
		x_i^{N, [n+1]} = \frac{ K_{_X} \big( x_{i + 1/2}^{N, [n]} \big) - K_{_X} \big( x_{i - 1/2}^{N, [n]} \big) }{ F_{_X} \big( x_{i + 1/2}^{N, [n]} \big) - F_{_X} \big( x_{i - 1/2}^{N, [n]} \big) }, \qquad i = 1, \dots, N.
	\end{aligned}
\end{equation}
In the seminal paper of \cite{kieffer1982exponential}, it has been shown that $\big( x^{[n]} \big)_{n \geq 1}$ converges exponentially fast toward $x$, the optimal quantizer, when the density $\varphi$ of $X$ is $\log$-concave and not piecewise affine. Numerical optimizations can be made in order to increase the rate of convergence to the optimal quantizer such as fixed point search acceleration, for example the Anderson acceleration (see \cite{anderson1965iterative} for the original paper and \cite{walker2011anderson} for details on the procedure).

Of course, other algorithms exist, such as the Newton Raphson zero search procedure or its variant the Levenberg–Marquardt algorithm which are deterministic procedures as well if the density, the first partial moment and the cumulative distribution function of $X$ are known. Additionally, we can cite stochastic procedures such as the CLVQ procedure (Competitive Learning Vector Quantization) which is a zero search stochastic gradient and the randomized version of the Lloyd's method I. For more details, the reader can refer to \cite{pages2018numerical,pages2016pointwise}.

Once the algorithm \eqref{RH:LloydAlgo} has been converging, we have at hand the quadratic optimal quantizer $\widehat X^N$ of $X$ and its associated probabilities given by
\begin{equation}\label{RH:probaOptimalQuantizer}
	\Prob \big( \widehat X^N = x_i^n \big) = F_{_X} \big( x_{i + 1/2}^{N} \big) - F_{_X} \big( x_{i - 1/2}^{N} \big), \qquad i = 1, \dots, n.
\end{equation}

%% file: RandomHeston.bbl
\begin{thebibliography}{RMKP17}

\bibitem[Alf05]{alfonsi2005discretization}
Aur{\'e}lien Alfonsi.
\newblock On the discretization schemes for the {CIR} (and bessel squared)
  processes.
\newblock {\em Monte Carlo Methods and Applications mcma}, 11(4):355--384,
  2005.

\bibitem[AMST07]{littletrapheston}
Hansj{\"o}rg Albrecher, {Philipp Arnold} Mayer, Wim Schoutens, and Jurgen
  Tistaert.
\newblock The little {H}eston trap.
\newblock {\em Wilmott}, (1):83--92, 2007.

\bibitem[And65]{anderson1965iterative}
Donald~G Anderson.
\newblock Iterative procedures for nonlinear integral equations.
\newblock {\em Journal of the ACM}, 12(4):547--560, 1965.

\bibitem[And07]{andersen2007efficient}
Leif~BG Andersen.
\newblock Efficient simulation of the heston stochastic volatility model.
\newblock {\em SSRN Electronic Journal}, 2007.

\bibitem[BP03]{bally2003quantization}
Vlad Bally and Gilles Pag{\`e}s.
\newblock A quantization algorithm for solving multidimensional discrete-time
  optimal stopping problems.
\newblock {\em Bernoulli}, 9(6):1003--1049, 2003.

\bibitem[BPP05]{printems2005quantization}
Vlad Bally, Gilles Pag{\`e}s, and Jacques Printems.
\newblock A quantization tree method for pricing and hedging multi-dimensional
  american options.
\newblock {\em Mathematical Finance}, 15(1):119--168, 2005.

\bibitem[CFG17]{callegaro2017pricing}
Giorgia Callegaro, Lucio Fiorin, and Martino Grasselli.
\newblock Pricing via recursive quantization in stochastic volatility models.
\newblock {\em Quantitative Finance}, 17(6):855--872, 2017.

\bibitem[CFG18]{callegaro2017american}
Giorgia Callegaro, Lucio Fiorin, and Martino Grasselli.
\newblock American quantized calibration in stochastic volatility.
\newblock {\em Risk Magazine}, 2018.

\bibitem[CGP18]{giorgia2018fast}
Giorgia Callegaro, Martino Grasselli, and Gilles Pag{\`e}s.
\newblock Fast hybrid schemes for fractional riccati equations (rough is not so
  tough).
\newblock {\em arXiv preprint arXiv:1805.12587}, 2018.

\bibitem[CIJR05]{cox2005theory}
John~C Cox, Jonathan~E Ingersoll~Jr, and Stephen~A Ross.
\newblock A theory of the term structure of interest rates.
\newblock In {\em Theory of Valuation}, pages 129--164. World Scientific, 2005.

\bibitem[CM99]{carr1999option}
Peter Carr and Dilip Madan.
\newblock Option valuation using the fast fourier transform.
\newblock {\em Journal of computational finance}, 2(4):61--73, 1999.

\bibitem[FSP18]{abbas2018product}
Lucio Fiorin, Abass Sagna, and Gilles Pag{\`e}s.
\newblock Product markovian quantization of a diffusion process with
  applications to finance.
\newblock {\em Methodology and Computing in Applied Probability}, pages 1--32,
  2018.

\bibitem[Gat11]{gatheral2011volatility}
Jim Gatheral.
\newblock {\em The volatility surface: a practitioner's guide}, volume 357.
\newblock John Wiley \& Sons, 2011.

\bibitem[GJR18]{gatheral2018volatility}
Jim Gatheral, Thibault Jaisson, and Mathieu Rosenbaum.
\newblock Volatility is rough.
\newblock {\em Quantitative Finance}, 18(6):933--949, 2018.

\bibitem[GJRS18]{guennoun2018asymptotic}
Hamza Guennoun, Antoine Jacquier, Patrick Roome, and Fangwei Shi.
\newblock Asymptotic behavior of the fractional heston model.
\newblock {\em SIAM Journal on Financial Mathematics}, 9(3):1017--1045, 2018.

\bibitem[GL00]{graf2000foundations}
Siegfried Graf and Harald Luschgy.
\newblock {\em Foundations of Quantization for Probability Distributions}.
\newblock Springer-Verlag, Berlin, Heidelberg, 2000.

\bibitem[GR09]{gauthier2009fitting}
Pierre Gauthier and Pierre-Yves~Henri Rivaille.
\newblock Fitting the smile, smart parameters for sabr and heston.
\newblock {\em SSRN Electronic Journal}, 2009.

\bibitem[GR19]{gatheral2019rational}
Jim Gatheral and Rados Radoicic.
\newblock Rational approximation of the rough heston solution.
\newblock {\em International Journal of Theoretical and Applied Finance},
  22(3):1950010, 2019.

\bibitem[Hes93]{heston1993closed}
Steven~L Heston.
\newblock A closed-form solution for options with stochastic volatility with
  applications to bond and currency options.
\newblock {\em The review of financial studies}, 6(2):327--343, 1993.

\bibitem[IW81]{ikeda1981stochastic}
Nobuyuki Ikeda and Shinzo Watanabe.
\newblock {\em Stochastic differential equations and diffusion processes},
  volume~24.
\newblock North Holland, 1981.

\bibitem[JR16]{jaisson2016rough}
Thibault Jaisson and Mathieu Rosenbaum.
\newblock Rough fractional diffusions as scaling limits of nearly unstable
  heavy tailed hawkes processes.
\newblock {\em The Annals of Applied Probability}, 26(5):2860--2882, 2016.

\bibitem[JS17]{jacquier2017randomized}
Antoine Jacquier and Fangwei Shi.
\newblock The randomized heston model.
\newblock {\em SIAM Journal on Financial Mathematics}, 10(1):89--129, 2017.

\bibitem[Kie82]{kieffer1982exponential}
John~C Kieffer.
\newblock Exponential rate of convergence for lloyd's method i.
\newblock {\em IEEE Transactions on Information Theory}, 28(2):205--210, 1982.

\bibitem[LL11]{lamberton2011introduction}
Damien Lamberton and Bernard Lapeyre.
\newblock {\em Introduction to stochastic calculus applied to finance}.
\newblock Chapman and Hall/CRC, 2011.

\bibitem[Llo82]{lloyd1982least}
Stuart Lloyd.
\newblock Least squares quantization in pcm.
\newblock {\em IEEE transactions on information theory}, 28(2):129--137, 1982.

\bibitem[MRKP18]{mcwalter2018recursive}
Thomas~A McWalter, Ralph Rudd, J{\"o}rg Kienitz, and Eckhard Platen.
\newblock Recursive marginal quantization of higher-order schemes.
\newblock {\em Quantitative Finance}, 18(4):693--706, 2018.

\bibitem[NM65]{nelder1965simplex}
John~A Nelder and Roger Mead.
\newblock A simplex method for function minimization.
\newblock {\em The computer journal}, 7(4):308--313, 1965.

\bibitem[Pag18]{pages2018numerical}
Gilles Pag{\`e}s.
\newblock {\em Numerical Probability: An Introduction with Applications to
  Finance}.
\newblock Springer, 2018.

\bibitem[PP09]{pages2009approximation}
Gilles Pag{\`e}s and Fabien Panloup.
\newblock Approximation of the distribution of a stationary markov process with
  application to option pricing.
\newblock {\em Bernoulli}, 15(1):146--177, 2009.

\bibitem[PPP04]{pages2004optimal}
Gilles Pag{\`e}s, Huy{\^e}n Pham, and Jacques Printems.
\newblock {\em Optimal Quantization Methods and Applications to Numerical
  Problems in Finance}, pages 253--297.
\newblock Birkh{\"a}user Boston, 2004.

\bibitem[PS15]{pages2015recursive}
Gilles Pag{\`e}s and Abass Sagna.
\newblock Recursive marginal quantization of the euler scheme of a diffusion
  process.
\newblock {\em Applied Mathematical Finance}, 22(5):463--498, 2015.

\bibitem[PS18]{sagna2018general}
Gilles Pag{\`e}s and Abass Sagna.
\newblock Weak and strong error analysis of recursive quantization: a general
  approach with an application to jump diffusions.
\newblock {\em arXiv preprint arXiv:1808.09755}, 2018.

\bibitem[PY16]{pages2016pointwise}
Gilles Pag{\`e}s and Jun Yu.
\newblock Pointwise convergence of the lloyd i algorithm in higher dimension.
\newblock {\em SIAM Journal on Control and Optimization}, 54(5):2354--2382,
  2016.

\bibitem[RMKP17]{rudd2017fast}
Ralph Rudd, Thomas McWalter, J{\"o}rg Kienitz, and Eckhard Platen.
\newblock Fast quantization of stochastic volatility models.
\newblock {\em SSRN Electronic Journal}, 2017.

\bibitem[Sag10]{sagna2010pricing}
Abass Sagna.
\newblock Pricing of barrier options by marginal functional quantization.
\newblock {\em Monte Carlo Methods and Applications}, 17(4):371--398, 2010.

\bibitem[SST04]{SchoutensWim2004}
Wim Schoutens, Erwin Simons, and Jurgen Tistaert.
\newblock A perfect calibration! now what?
\newblock {\em Wilmott Magazine}, 2004.

\bibitem[WN11]{walker2011anderson}
Homer~F Walker and Peng Ni.
\newblock Anderson acceleration for fixed-point iterations.
\newblock {\em SIAM Journal on Numerical Analysis}, 49(4):1715--1735, 2011.

\bibitem[Zad82]{zador1982asymptotic}
Paul Zador.
\newblock Asymptotic quantization error of continuous signals and the
  quantization dimension.
\newblock {\em IEEE Transactions on Information Theory}, 28(2):139--149, 1982.

\end{thebibliography}
